\documentclass[11pt]{article}
\usepackage{lgrind}
\usepackage{cmap}
\usepackage[T1]{fontenc}
\usepackage{graphicx}
\usepackage{amsthm}
\usepackage{amsmath}
\usepackage[linesnumbered,algonl,boxed, ruled]{algorithm2e}
\usepackage{amssymb}
\usepackage{color}
\usepackage{mdframed, framed}
\usepackage{caption}
\usepackage{lineno}

\newcommand{\R}{{\cal R}}
\newcommand{\OO}{{\cal O}}
\newcommand{\KK}{\mathcal{K}}
\newcommand{\LL}{\mathcal{L}}

\newcommand{\FF}{\mathcal{F}}
\newcommand{\HH}{\mathcal{H}}
\newcommand{\MM}{\mathcal{M}}
\newcommand{\GG}{\mathcal{G}}

\newcommand{\VV}{\mathcal{V}}
\newcommand{\SSS}{\mathcal{S}}

\newcommand{\CON}{\mbox{Con}}
\newcommand{\JJ}{\mathcal{J}}

\newcommand{\REL}{\mathbf{REL}}
\newcommand{\PSR}{\mathbf{PSR}}

\newtheorem{theorem}[figure]{Theorem}

\newtheorem{definition}[figure]{Definition}
\newtheorem{property}[figure]{Property}
\newtheorem{lemma}[figure]{Lemma}

\newtheorem{observation}[figure]{Observation}

\title{A Polynomial Time Algorithm for Finding Area-Universal
\\ Rectangular Layouts}

\date{}
\author{
Jiun-Jie Wang\\
State University of New York  at Buffalo, Buffalo, NY 14260, USA.\\
Email: jiunjiew@buffalo.edu
}
\begin{document}
\maketitle

\begin{abstract}
A rectangular layout $\LL$ is a rectangle partitioned into
disjoint smaller rectangles so that no four smaller rectangles meet at
the same point. Rectangular layouts were originally used as floorplans in VLSI design
to represent VLSI chip layouts. More recently, they are used in graph
drawing as rectangular cartograms. In these applications, an area
$a(r)$ is assigned to each rectangle $r$, and the actual area of
$r$ in $\LL$ is required to be $a(r)$. Moreover, some applications
require that we use combinatorially equivalent rectangular layouts
to represent multiple area assignment functions. $\LL$ is called
{\em area-universal} if any area assignment to its rectangles can be
realized by a layout that is combinatorially equivalent to $\LL$.

A basic question in this area is to determine if a given plane graph $G$
has an area-universal rectangular layout or not. A fixed-parameter-tractable
algorithm for solving this problem was obtained in \cite{EMSV12}.
Their algorithm takes $O(2^{O(K^2)}n^{O(1)})$ time (where $K$ is the
maximum number of degree 4 vertices in any minimal separation component),
which is exponential time in general case. It is an open problem to find a
true polynomial time algorithm for solving this problem.
In this paper, we describe such a polynomial time algorithm.
Our algorithm is based on new studies of properties of area-universal
layouts. The polynomial run time is achieved by exploring
their connections to the {\em regular edge labeling} construction.

\end{abstract}


\section{Introduction}\label{sec:Intro-universal}
A {\em rectangular layout} $\LL$ is a partition of a rectangle $R$
into a set $R(\LL)=\{ r_1, \ldots , r_n\}$ of disjoint smaller
rectangles by vertical and horizontal line segments so that no
four smaller rectangles meet at the same point.
An {\em area assignment function} of a rectangular layout $\LL$
is a function $a: R(\LL) \rightarrow \mathbb{R}^+$. We say $\LL$
is a {\em rectangular cartogram} for $a$ if the area of each
$r_i\in R(\LL)$ equals to $a(r_i)$. We also say $\LL$
{\em realizes} the area assignment function $a$.

Rectangular cartograms were introduced in \cite{Ra34} to display
certain numerical quantities associated with geographic regions.
Each rectangle $r_i$ represents a geographic region. Two regions are
geographically adjacent if and only if their corresponding rectangles
share a common boundary in $\LL$. The areas of the rectangles
represent the numeric values being displayed by the cartogram.

In some applications, several sets of numerical data must be displayed
as cartograms of the same set of geographic regions. For example,
three figures in \cite{Ra34} are the cartograms of land area,
population, and wealth within the United States. In such cases, we
wish to use cartograms whose underlying rectangular layouts are
{\em combinatorially equivalent} (to be defined later).
Fig \ref{fig:layout} (1) and (2) show two combinatorially
equivalent layouts with different area assignments.
The following notion was introduced in \cite{EMSV12}.

\begin{figure}[t]
\begin{center}
\includegraphics[width=0.45\textwidth, angle =0]{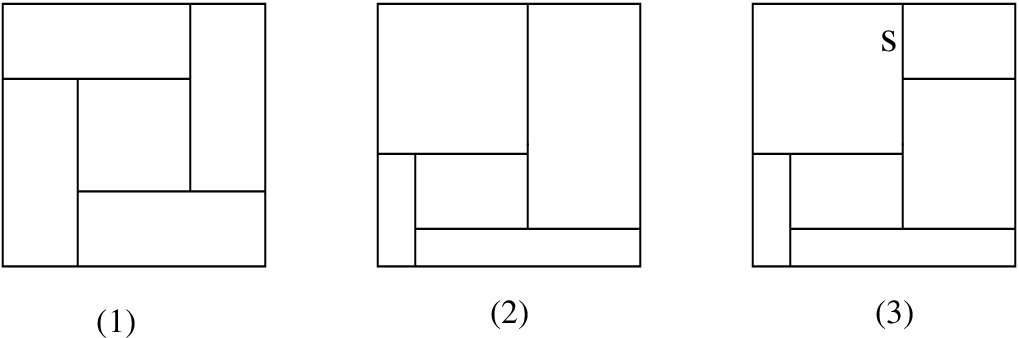}
  \centering
\caption{Examples of rectangular layout: (1) and (2) are
two combinatorially equivalent layouts with different
area assignments. Both are area-universal layouts.
(3) A layout that is not area-universal.}
\label{fig:layout}
\end{center}
\end{figure}
\begin{definition}
A rectangular layout $\LL$ is {\em area-universal} if any
area assignment function $a$ of $\LL$ can be realized by a
rectangular layout that is {\em combinatorially equivalent} to $\LL$.
\end{definition}


A natural question is: which layouts are area-universal?
A nice characterization of area-universal rectangular layouts was
discovered in \cite{EMSV12}:

\begin{theorem}$\cite{EMSV12}$\label{thm:universal}
A rectangular layout $\LL$ is area-universal if and only
if every maximal line segment in $\LL$ is a side of at
least one rectangle in $\LL$.
(A maximal line segment is a line segment in $\LL$ that cannot
be extended without crossing other line segments in $\LL$.)
\end{theorem}

In Fig \ref{fig:layout}, the layouts (1) and (2) are area-universal,
but the layout (3) is not. (The maximal vertical line segment $s$
is not a side of any rectangle.)

For a plane graph $G$, we say a rectangular layout $\LL$ {\em represents}
$G$ if the following hold: (1) The set of smaller rectangles of $\LL$
one-to-one corresponds to the set of vertices of $G$; and
(2) two vertices $u$ and $v$ are adjacent in $G$ if and only if their
corresponding rectangles in $\LL$ share a common boundary. In other words,
if $\LL$ represents $G$, then $G$ is the dual graph of small rectangles in $\LL$.

Area-universal rectangular layout representations of graphs are useful
in other fields \cite{Mu08}. In VLSI design, for example \cite{YS95},
the rectangles in $\LL$ represent circuit components, and the common
boundary between rectangles in $\LL$ model the adjacency requirements
between components. In early VLSI design stage, the chip areas of circuit
components are not known yet. Thus, at this stage, only the relative
positions of the components are considered. At later design stages, the
areas of the components (namely, the rectangles in $\LL$) are specified.
An area-universal layout $\LL$ enables the realization of the area
assignments specified at later design stages. Thus, the ability of
finding an area-universal layout at the early
design stage will greatly simplify the design process at later stages.
The applications of rectangular layouts and cartograms in building design
and in tree-map visualization can be found in \cite{EM79,BHW00}.
Heuristic algorithms for computing the coordinates of a rectangular layout
that realizes a given area assignment function were presented in \cite{WKC88,KS07}.

A plane graph $G$ may have many rectangular layouts. Some of them
may be area-universal, while the others are not. Not every plane graph
has an area-universal layout. In \cite{Ri87}, Rinsma described an
outerplanar graph $G$ and an area assignment to its vertices such that
no rectangular layout realizes the area assignment.
Thus it is important to determine if $G$ has an area-universal layout
or not. Based on Theorem \ref{thm:universal}, Eppstein et al.
$\cite{EMSV12}$ described an algorithm that finds an area-universal
layout for $G$ if one exists. Their algorithm takes $O(2^{O(K^2)}n^{O(1)})$
time, where $K$ is the maximum number of degree 4 vertices in any minimal
separation component. For a fixed $K$, the algorithm runs in polynomial
time. However, their algorithm takes exponential time in general case.

In this paper, we describe the first polynomial-time algorithm for
solving this problem. Our algorithm is based on studies of properties of
area-universal layouts and their connection to the {\em regular edge
labeling} construction.
The paper is organized as follows. In \S \ref{sec:pre-universal}, we introduce
basic definitions and preliminary results. \S \ref{sec:outline} outlines
a Face-Addition algorithm with exponential time that determines if $G$
has an area-universal rectangular layout. \S \ref{sec:concept} introduces the
concepts of forbidden pairs, $\GG$-pairs and $\MM$-triples that are extensively used
in our algorithm. In \S \ref{sec:impl}, we describe how to convert
the Face-Addition algorithm with exponential time to an algorithm with
polynomial time.

\section{Preliminaries}\label{sec:pre-universal}

In this section, we give definitions and important preliminary results.
Definitions not mentioned here are standard. A graph $G=(V,E)$ is called
{\em planar} if it can be drawn on the plane with no edge
crossings. Such a drawing is called a {\em plane embedding} of $G$.
A {\em plane graph} is a planar graph with a fixed plane embedding.
A plane embedding of $G$ divides the plane into a number of connected
regions. Each region is called a {\em face}. The unbounded region is called
the {\em exterior face}. The other regions are called {\em interior faces}.
The vertices and edges on the exterior face are called {\em exterior
vertices and edges}. Other vertices and edges are called {\em interior vertices
and edges}. We use {\em cw} and {\em ccw} as the abbreviation of {\em clockwise}
and {\em counterclockwise}, respectively.

For a simple path $P=\{v_1, v_2, \cdots, v_p\}$ of $G$, the {\em length}
of $P$ is the number of edges in $P$. $P$ is called \emph{chord-free} if
for any two vertices $v_i, v_j$ with $|i-j|>1$, the edge $(v_i, v_j)\notin E$.
A {\em triangle} of a plane graph $G$ is a cycle $C$ with three edges.
$C$ divides the plane into its interior and exterior regions.
A {\em separating triangle} is a triangle in $G$ such that
there are vertices in both the interior and the exterior of $C$.

When discussing the rectangular layout $\LL$ of a plane graph $G$, we can simplify
the problem as follows. Let $a,b,c,d$ be the four designated exterior vertices of $G$
that correspond to the four rectangles in $\LL$ located at the southwest,
northwest, northeast and southeast corners, respectively. Let the
{\em extended graph} $G_{ext}$ be the graph obtained from $G$ as follows:

\begin{enumerate}
\item Add four vertices $v_W,v_N,v_E,v_S$ and four edges
$(v_W,v_N),(v_N,v_E),(v_E,v_S),(v_S,v_W)$ into $G_{ext}$.
\item
Connect $v_W$ to every vertex of $G$ on the exterior face between $a$ and $b$
in cw order.
Connect $v_N$ to every vertex of $G$ on the exterior face between $b$ and $c$
in cw order.
Connect $v_E$ to every vertex of $G$ on the exterior face between $c$ and $d$
in cw order.
Connect $v_S$ to every vertex of $G$ on the exterior face between $d$ and $a$
in cw order.
\end{enumerate}

See Figs \ref{fig:REL} (1) and (2) for an example.
It is well known \cite{KK85} that
$G$ has a rectangular layout $\LL$ if and only if $G_{ext}$ has a
rectangular layout $\LL_{ext}$, where the rectangles corresponding to
$v_W,v_N,v_E,v_S$ are located at the west, north, east and south boundary
of $\LL_{ext}$, respectively.
Not every plane graph has rectangular layouts. The following
theorem characterizes the plane graphs with rectangular layouts.

\begin{figure}[t]
\begin{center}
\includegraphics[width=0.8\textwidth, angle =0]{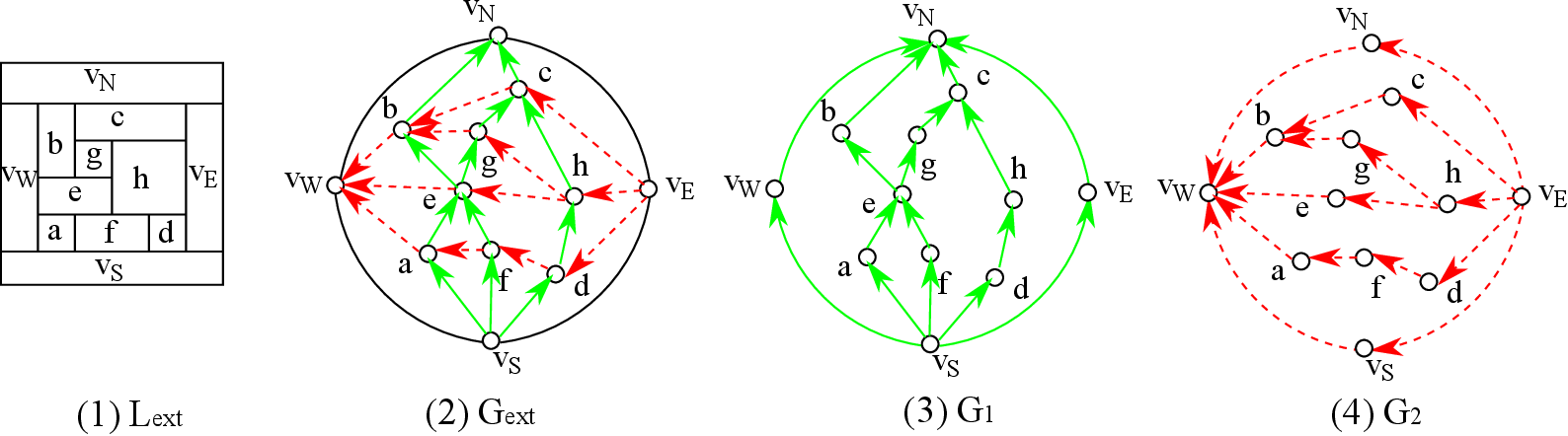}
  \centering
\caption{Examples of rectangular layout and $\REL$. (1) Rectangular
layout $\LL_{ext}$; (2) The graph corresponding to $\LL_{ext}$ with
an $\REL$ $\R=\{T_1,T_2\}$; (3) the graph $G_1$ of $\R$; (4) the
graph $G_2$ of $\R$.}
\label{fig:REL}
\end{center}
\end{figure}

\begin{theorem}\label{thm:proper} \cite{KK85}
A plane graph $G$ has a rectangular layout $\LL$ with
four rectangles on its boundary if and only if:
\begin{enumerate}
\item Every interior face of $G$ is a triangle and the exterior face
of $G$ is a quadrangle; and \item $G$ has no separating triangles.
\end{enumerate}
\end{theorem}

A plane graph that satisfies the conditions in Theorem \ref{thm:proper}
is called a {\em proper triangular plane graph}. From now on we
only consider such graphs.

Our algorithm relies heavily on the concept of the \emph{regular edge
labeling} ($\REL$) introduced in $\cite{He93}$. $\REL$s have also been
studied by Fusy \cite{Fu06,Fu09}, who refers them as {\em transversal
structures}. $\REL$ are closely related to several other edge
coloring structures of planar graphs that can be used to describe
straight line embeddings of orthogonal polyhedra \cite{Ep10,EM10}.

\begin{definition}\label{def:REL}
Let $G$ be a proper triangular plane graph.
A regular edge labeling $\REL$ $\R=\{T_1, T_2\}$ of
$G$ is a partition of the interior edges of $G$ into two
subsets $T_1, T_2$ of directed edges such that:

\begin{itemize}
\item For each interior vertex $v$, the edges incident to $v$
appear in ccw order around $v$ as follows: a set of
edges in $T_1$ leaving $v$; a set of edges in $T_2$ leaving $v$;
a set of edges in $T_1$ entering $v$; a set of edges in $T_2$ entering $v$.
(Each of the four sets contains at least one edge.)

\item Let $v_N, v_W, v_S, v_E$ be the four exterior vertices in
ccw order. All interior edges incident to $v_N$ are
in $T_1$ and entering $v_N$. All interior edges incident to
$v_W$ are in $T_2$ and entering $v_W$. All interior edges incident
to $v_S$ are in $T_1$ and leaving $v_S$. All interior edges
incident to $v_E$ are in $T_2$ and leaving $v_E$.
\end{itemize}
\end{definition}

Fig \ref{fig:REL} (2) shows an example of $\REL$. (The green solid
lines are edges in $T_1$. The red dashed lines are edges in $T_2$.)
It is well known that every proper triangular plane graph $G$ has a
$\REL$, which can be found in linear time \cite{He93,KH97}. Moreover,
from a $\REL$ of $G$, we can construct a rectangular layout $\LL$
of $G$ in linear time \cite{He93,KH97}. Conversely, if we have a
rectangular layout $\LL$ for $G$, we can easily obtain a $\REL$ $\R$
of $G$ as follows. For each interior edge $e=(u,v)$ in $G$, we label
and direct $e$ according to the following rules.  Let $r_u$ and $r_v$
be the rectangle in $\LL$ corresponding to $u$ and $v$ respectively.

\begin{itemize}
\item  If $r_u$ is located below $r_v$ in $\LL$, the edge $e$
is in $T_1$ and directed from $u$ to $v$.
\item   If $r_u$ is located to the right of $r_v$ in $\LL$,
the edge $e$ is in $T_2$ and directed from $u$ to $v$.
\end{itemize}

The $\REL$ $\R$ obtained as above is called the {\em $\REL$ derived from $\LL$}.
(See Fig \ref{fig:REL} (1) and (2)).

\begin{definition}
Let $\LL_1$ and $\LL_2$ be two rectangular layouts of a proper triangular
plane graph $G$. We say $\LL_1$ and $\LL_2$ are {\rm combinatorially
equivalent} if the $\REL$s of $G$ derived from $\LL_1$ and
from $\LL_2$  are identical.
\end{definition}

Thus, the $\REL$s of $G$ one-to-one correspond to
the combinatorially equivalent rectangular layouts of $G$.
We can obtain two directed subgraphs $G_1$ and $G_2$ of $G$ from
an $\REL$ $\R=\{T_1, T_2\}$ as follows.

\begin{itemize}
\item The vertex set of $G_1$ is $V$. The edge set of $G_1$ consists of
the edges in $T_1$ with direction in $T_1$, and the four exterior edges
directed as: $v_S \rightarrow v_W, v_S \rightarrow v_E,
v_W \rightarrow v_N, v_E \rightarrow v_N$.
\item The vertex set of $G_2$ is $V$. The edge set of $G_2$ consists of
the edges in $T_2$ with direction in $T_2$, and the four exterior edges
directed as: $v_S \rightarrow v_W, v_N \rightarrow v_W,
v_E \rightarrow v_S, v_E \rightarrow v_N$.
\end{itemize}

Fig \ref{fig:REL} (3) and (4) show the graph $G_1$ and $G_2$
for the $\REL$ shown in Fig \ref{fig:REL} (2).
For each face $f_1$ in $G_1$, the boundary of $f_1$ consists of two
directed paths. They are called the two {\em sides} of $f_1$.
Each side of $f_1$ contains at least two edges. Similar properties
hold for the faces in $G_2$ \cite{Fu06,Fu09,He93,KH97}.

\begin{definition}
A $\REL$ $\R=\{T_1, T_2\}$ of $G$ is called {\em slant}
if for every face $f$ in either $G_1$ or $G_2$, at least one
side of $f$ contains exactly two directed edges.
\end{definition}

Theorem \ref{thm:universal} characterizes the area-universal layouts in terms
of maximal line segments in $\LL$. The following lemma characterizes
area-universal layouts in term of the $\REL$ derived from $\LL$.

\begin{lemma}\label{lemma:slant}
A rectangular layout $\LL$ is area-universal if and
only if the $\REL$ $\R$ derived from $\LL$ is slant.
\end{lemma}
\begin{proof}
Note that each face in $G_1$ ($G_2$, respectively) corresponds
to a maximal vertical (horizontal, respectively) line
segment in $\LL$. (In the graph $G_1$ in Fig
\ref{fig:REL} (3), the face $f_1$ with the vertices $f,e,g,c,h$
corresponds to the vertical line segment that is on the left side
of the rectangle $h$ in Fig \ref{fig:REL} (1)).

Assume $\LL$ is area-universal. Consider a face $f$ in $G_1$. Let
$l_f$ be the maximal vertical line segment in $\LL$ corresponding to
$f$. Since $\LL$ is area-universal, $l_f$ is a side of a rectangle
$r$ in $\LL$. Without loss of generality, assume $r$ is to the left
of $l_f$. Then the left side of the face $f$ consists
of exactly two edges. Thus $G_1$ satisfies the slant property.
Similarly, we can show $G_2$ also satisfies the slant property.

Conversely, assume $\R$ is a slant $\REL$. The above argument can be
reversed to show that $\LL$ is area-universal.
\end{proof}

The $\REL$ shown in Fig \ref{fig:REL} (2) is not slant because
the slant property fails for one $G_2$ face.
So the corresponding layout shown in Fig \ref{fig:REL} (1)
is not area-universal. By Lemma \ref{lemma:slant}, the problem of
finding an area-universal layout for $G$ is the same as the problem
of finding a slant $\REL$ for $G$. From now on, we consider the
latter problem and $G$ always denotes a proper triangular plane graph.

\section{Face-Addition Algorithm with Exponential Time}\label{sec:outline}

In this section, we outline a Face-Addition procedure that generates
a slant $\REL$ $\R =\{T_1,T_2\}$ of $G$ through a sequence of steps.
The procedure starts from the directed path consisting of two edges
$v_S \rightarrow v_E \rightarrow v_N$. Each step maintains a partial
slant $\REL$ of $G$. During a step, a face $f$ of $G_1$ is added to the
current graph, resulting in a larger partial slant $\REL$. When $f$ is
added, its right side is already in the current graph. The edges on the
left side of $f$ are placed in $T_1$ and directed upward. The edges
of $G$ in the interior of $f$ are placed in $T_2$ and directed to the
left. The process ends when the left boundary $v_S \rightarrow v_W
\rightarrow v_N$ is reached. With this informal description in mind,
we first introduce a few definitions. Then we will formally
describe the Face-Addition algorithm (which takes exponential time).

\begin{figure}[t]
\begin{center}
\includegraphics[width=0.4\textwidth, angle =0]{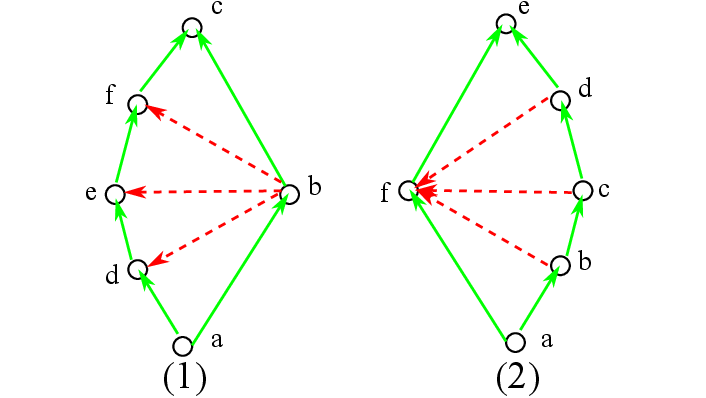}
\centering
\caption{
(1)  a fan $\FF(a, b, c)$ has the back boundary $(a, b, c)$ and the front boundary $(a, d, e, f, c)$;
(2)  a mirror fan $\MM(a, f, e)$ has the back boundary $(a, b, c, d, e)$ and the front boundary $(a, f, e)$.}
\label{fig:gadget}
\end{center}
\end{figure}

Consider a face $f$ of $G_1$ added during the above procedure.
Because we want to generate a slant $\REL$ $\R$, at least one side of
$f$ must be a path of length 2. This motivates the following definition.
Figs \ref{fig:gadget} (1) and (2)
show examples of a fan and a mirror fan, respectively.

\begin{definition}\label{def:gadget}
Let $v_l,v_m,v_h$ be three vertices of $G$ such that
$v_l$ and $v_h$ are two neighbors of $v_m$ and $(v_l,v_h) \notin E$.
Let $P_{cw}$ be the path consisting of the neighbors $\{ v_l,v_1,\ldots,v_p,v_h\}$
of $v_m$ in cw order between $v_l$ and $v_h$.
Let $P_{ccw}$ be the path consisting of the neighbors $\{ v_l,u_1,\ldots,u_q,v_h\}$
of $v_m$ in ccw order between $v_l$ and $v_h$. Note that since $G$ has no
separating triangles, both $P_{cw}$ and $P_{ccw}$ are chord-free.
\begin{enumerate}

\item The directed and labeled subgraph
of $G$ induced by the vertices $v_l,v_m,v_h,v_1,\ldots,v_p$
is called the {\em fan at} $\{v_l, v_m, v_h\}$
and denoted by $\FF(v_l, v_m, v_h)$, or simply $g$.

\begin{itemize}
\item The \emph{front boundary} of $g$, denoted by $\alpha(g)$,
consists of the edges in $P_{cw}$ directed from $v_l$ to $v_h$ in cw order.
The edges in $\alpha(g)$ are colored green.

\item The \emph{back boundary} of $g$, denoted by $\beta(g)$,
consists of two directed edges $v_l \rightarrow v_m$ and $v_m \rightarrow  v_h$.
The edges in $\beta(g)$ are colored green.

\item The \emph{inner edges} of $g$, denote by $\gamma(g)$, are the edges
between $v_m$ and the vertices $v\neq v_l,v_h$ that are on the path
$P_{cw}$. The inner edges are colored red and directed away from $v_m$.
\end{itemize}

\item The directed and labeled
subgraph of $G$ induced by the vertices $v_l,v_m,v_h,u_1,\ldots,u_q$
is called the {\em mirror fan at} $\{v_l, v_m, v_h\}$
and denoted by $\MM(v_l, v_m, v_h)$, or simply $g$.

\begin{itemize}
\item The \emph{front boundary} of $g$, denoted by $\alpha(g)$,
consists of two directed edges $v_l\rightarrow v_m$ to $v_m\rightarrow v_h$.
The edges in $\alpha(g)$ are colored green.

\item The \emph{back boundary} of $g$, denoted by $\beta(g)$, consists
of the edges in $P_{ccw}$ directed from $v_l$ to $v_h$ in ccw order.
The edges in $\beta(g)$ are colored green.

\item The \emph{inner edges} of $g$, denote by $\gamma(g)$, are the edges
between $v_m$ and the vertices $v\neq v_l,v_h$ that are on the path
$P_{ccw}$. The inner edges are colored red and directed into $v_m$.
\end{itemize}
\end{enumerate}
\end{definition}

Both $\FF(v_l, v_m, v_h)$ and $\MM(v_l, v_m, v_h)$ are
called a {\em gadget at $v_l,v_m, v_h$}. We use $g(v_l, v_m, v_h)$ to
denote either of them. The vertices other than $v_l$ and $v_h$
are called the {\em internal vertices} of the gadget.
If a gadget has only one inner edge, it can be called either
a fan or a mirror fan. For consistency, we call it a fan.
We use $g_0=\FF(v_S,v_E,v_N)$ to denote the {\em initial fan}, and
$g_T=\MM(v_S,v_W,v_N)$ to denote the {\em final mirror fan}.
The following observation is clear:

\begin{observation}\label{obs:slant}
For a slant $\REL$, each face $f$ of $G_1$ is a gadget of $G$.
\end{observation}

The $\REL$ shown in Fig \ref{fig:REL} (2) is generated by adding the
gadgets: $\FF(v_S,v_E,v_N)$, $\FF(v_S,d,h)$,
$\FF(f,h,c)$, $\MM(e,b,v_N)$, $\FF(v_S,f,e)$, $\MM(v_S,v_W,v_N)$.
The following lemma is needed later.

\begin{lemma}\label{lemma:number}
The total number of gadgets in $G$ is at most $O(n^2)$.
\end{lemma}

\begin{proof}
Let $\mbox{deg}(v)$ denote the degree of the vertex $v$ in $G$. For each
$v$, there are at most $2\cdot\mbox{deg}(v)\cdot (\mbox{deg}(v)-3)$
gadgets with $v$ as its middle element. Thus the total number of gadgets
of $G$ is at most
$\sum_{v\in V}2\cdot\mbox{deg}(v) \cdot (\mbox{deg}(v)-3) = O(n^2)$.
\end{proof}

\begin{definition}\label{def:cut}
A \emph{cut} $C$ of $G$ is a directed path from $v_S$ to $v_N$ that is
the left boundary of the subgraph of $G$ generated during the
Face-Addition procedure. In particular,
$C_0= v_S\rightarrow v_E \rightarrow v_N$ denotes the {\em initial cut} and
$C_T= v_S\rightarrow v_W \rightarrow v_N$ denotes the {\em final cut}.
\end{definition}

Let $C$ be a cut of $G$. For any two vertices $v_1, v_2$ of $C$, $C(v_1, v_2)$
denotes the subpath of $C$ from $v_1$ to $v_2$. The two paths $C$ and $C_0$
enclose a region on the plane. Let $G_{|C}$ denote the subgraph
of $G$ induced by the vertices in this region (including its boundary).

Consider a cut $C$ generated by Face-Addition procedure
and a gadget $g=g(v_l, v_m, v_h)$. In order for Face-Addition procedure
to add $g$ to $C$, the following conditions must be satisfied:
\begin{description}
\item[A1:] no internal vertices of $\alpha(g)$ are in $C$; and
\item[A2:] the back boundary $\beta(g)$ is contained in $C$; and
\item[A3:] $g$ is {\em valid} for $C$ (the meaning of {\em valid}
will be defined later).
\end{description}

If $g$ satisfies the conditions A1, A2 and A3, Face-Addition
procedure can add $g$ to the current graph $G_{|C}$ by {\em stitching}
$\beta(g)$ with the corresponding vertices on $C$. (Intuitively we are
{\em adding a face of $G_1$}.) Let $G_{|C} \otimes g$ denote the new subgraph
obtained by adding $g$ to $G_{|C}$. The new cut of $G_{|C} \otimes g$,
denoted by $C\otimes g$, is the concatenation of three subpaths
$C(v_S, v_l), \alpha(g), C(v_h, v_N)$.

The conditions A1 and A2 ensure that $C \otimes g$ is a cut. Any gadget
$g$ satisfying A1 and A2 can be added during a step while still maintaining
the slant property for $G_1$. However, adding such a $g$ may destroy the
slant property for $G_2$ faces. The condition A3 that $g$ is {\em valid}
for $C$ is to ensure the slant property for $G_2$ faces. (The $\REL$
shown in Fig \ref{fig:REL} (2) is not slant. This is because the gadget
$\FF(f,h,c)$ is not valid, as we will explain later.)
This condition will be discussed in \S \ref{sec:concept}.

After each iteration of Face-Addition procedure, the edges of
the current cut $C$ are always in $T_1$ and directed from $v_S$
to $v_N$. All $G_1$ faces $f_1$ in $G_{|C}$ are {\em complete}
(i.e. both sides of $f_1$ are in $G_{|C}$). Some $G_2$ faces in
$G_{|C}$ are complete. Some other $G_2$ faces $f_2$ in $G_{|C}$ are
{\em open}. (i.e. the two sides of $f_2$ are not completely in $G_{|C}$.)

\begin{definition}
Any subgraph $G_{|C}$ generated during the execution of
Face-Addition procedure is called a {\em partial slant}
$\REL$ of $G$, which satisfies the following conditions:
\begin{enumerate}
\item Every complete $G_1$ and $G_2$ face in $G_{|C}$ satisfies
the slant $\REL$ property.
\item For every open $G_2$ face $f$ in $G_{|C}$, at least one
side of $f$ has exactly one edge.
\end{enumerate}
\end{definition}

The intuitive meaning of a partial slant $\REL$ $G_{|C}$ is that it is
{\em potentially possible to grow} a complete slant $\REL$ of $G$ from
$G_{|C}$. The left boundary of a partial slant $\REL$ $\R$ is called
the {\em cut associated with} $\R$ and denoted by $C(\R)$.

\begin{definition}
\label{def:G-tilde}
\begin{enumerate}
\item $\PSR(G)$ denotes the set of all partial slant $\REL$s of $G$
that can be generated by Face-Addition procedure.
\item $\tilde{G} = \{ g~|~g \mbox{ is a gadget in a } \R \in \PSR(G)\}$.
\end{enumerate}
\end{definition}

Observe that every slant $\REL$ $\R$ of $G$ is in $\PSR(G)$. This is because
$\R$ is generated by adding a sequence of gadgets $g_1, \ldots g_T=\MM(v_S,v_W,v_N)$
to the initial gadget $g_0= \FF(v_S,v_E,v_N)$.
So if we choose this particular $g_i$ during the $i$th step, we will get $\R$
at the end. Thus $G$ has a slant $\REL$ if and only if $g_T \in \tilde{G}$.
Note that Face-Addition procedure works only if we know the correct
gadget addition sequence. Of course, we do not know such a sequence.
The Face-Addition algorithm, described in Algorithm
\ref{alg:face-addition}, generates all members in $\PSR(G)$.

\begin{algorithm}[htb]
\caption{Face-Addition algorithm with Exponential Time}
\label{alg:face-addition}

Initialize $\tilde{G}=\{ g_0\}$,
and $\PSR(G)=\{G_{|\alpha(g_0)}\}$\;

\Repeat{no such $g$ and $\R$ can be found}
{
Find a gadget $g$ of $G$ and an $\R\in \PSR(G)$ such that
the conditions A1, A2 and A3 are satisfied for $g$ and $C=C(\R)$\;

Add $g$ into $\tilde{G}$,
and add the partial slant $\REL$ $\R \otimes g$ into $\PSR(G)$\;
}

$G$ has a slant $\REL$ if and only if the final gadget $g_T \in \tilde{G}$\;
\end{algorithm}

Because $|\PSR(G)|$ can be exponentially large, Algorithm
\ref{alg:face-addition} takes exponential time.

\section{Forbidden Pairs, $\GG$-Pairs, $\MM$-Triples, Chains and Backbones}
\label{sec:concept}

In this section, we describe the conditions for adding a gadget to a
partial slant $\REL$ $\R \in \PSR(G)$, while still keeping the slant
$\REL$ property for $G_2$ faces. (In other words, the condition A3.)

\subsection{Forbidden Pairs}\label{sec:forbidden}

Consider a $\R \in \PSR(G)$ and its associated cut $C=C(\R)$.
Let $e$ be an edge of $C$. We use $\mbox{open-face}(e)$ to denote the
open $G_2$ face in $G_{|C}$ with $e$ as its open left boundary.
The {\em type} of $\mbox{open-face}(e)$ specifies the lengths of the
lower side $P_l$ and the upper side $P_u$ of $\mbox{open-face}(e)$:

\begin{itemize}
\item Type (1,1): $\mbox{length}(P_l)= 1$ and $\mbox{length}(P_u)= 1$.
\item Type (1,2): $\mbox{length}(P_l)= 1$ and $\mbox{length}(P_u)\geq 2$.
\item Type (2,1): $\mbox{length}(P_l)\geq 2$ and $\hbox{length}(P_u)= 1$.
\item Type (2,2): $\mbox{length}(P_l)\geq 2$ and $\mbox{length}(P_u)\geq 2$.
\end{itemize}

Note that the type of every open $G_2$ face in a partial slant
$\REL$ cannot be $(2,2)$. Based on the properties of $\REL$, we
have the following (see Fig \ref{fig:face-type}):

\begin{observation}
Let $\R \in \PSR(G)$ and $e$ be an edge on $C(\R)$.
\begin{itemize}
\item If $e$ is the last edge of $\alpha(g)$ of a fan or a mirror fan $g$,
the type of $\mbox{open-face}(e)$ is (2,1).
\item If $e$ is a middle edge of $\alpha(g)$ of a fan $g$,
the type of $\mbox{open-face}(e)$ is (1,1).
\item If $e$ is the first edge of $\alpha(g)$ of a fan or a mirror fan $g$,
the type of $\mbox{open-face}(e)$ is (1,2).
\end{itemize}
\end{observation}

\begin{figure}[ht]
\begin{center}
\includegraphics[width=0.40\textwidth, angle =0]{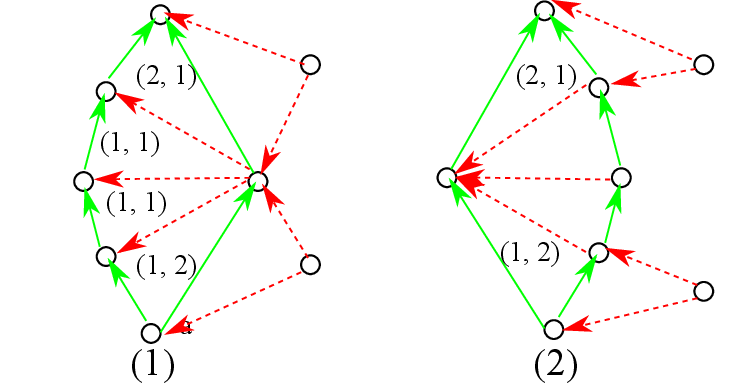}
  \centering
\caption{The types of open $G_2$ faces:
(1) Faces defined by edges on the front boundary of a fan;
(2) Faces defined by edges on the front boundary of a mirror fan.}
\label{fig:face-type}
\end{center}
\end{figure}
\vspace{-0.25in}

\begin{definition}\label{def:forbidden}
A pair $(g,g')$ of two gadgets of $G$ is called a {\em forbidden pair} if either
(1) the first edge of $\beta(g)$ is the last edge of $\alpha(g')$; or
(2) the last edge of $\beta(g)$ is the first edge of $\alpha(g')$.
\end{definition}

\begin{lemma}\label{lemma:forbidden}
If a partial $\REL$ $\R$ contains a forbidden pair $(g,g')$,
then $\R$ is not slant.
\end{lemma}

\begin{proof}
Case 1: Suppose the first edge $e_1$ of $\beta(g)$ is the last edge
of $\alpha(g')$ (see Fig \ref{fig:forbbiden} (1)). Let $e_2$ be the
first edge of $\alpha(g)$. The type of $\mbox{open-face}(e_1)$ is
$(2,1)$ (regardless of whether $g'$ is a fan or a mirror fan). Note
that $\mbox{open-face}(e_2)$ extends $\mbox{open-face}(e_1)$.
The length of the upper side of $\mbox{open-face}(e_1)$ is increased
by 1. Thus the type of $\mbox{open-face}(e_2)$ is $(2,2)$ and
the slant property for $G_2$ face fails.

Case 2: Suppose the last edge $e_1$ of $\beta(g)$ is the first edge
of $\alpha(g')$ (see Fig \ref{fig:forbbiden} (2)). Let $e_2$ be
the last edge of $\alpha(g)$. The type of $\mbox{open-face}(e_1)$
is $(1,2)$ (regardless of whether $g'$ is a fan or a mirror fan).
Note that $\mbox{open-face}(e_2)$ extends $\mbox{open-face}(e_1)$.
The length of the lower side of $\mbox{open-face}(e_1)$ is increased
by 1. Thus the type of $\mbox{open-face}(e_2)$ is $(2,2)$ and
the slant property for $G_2$ face fails.
\end{proof}

\begin{figure}[ht]
\begin{center}
\includegraphics[width=0.38\textwidth, angle =0]{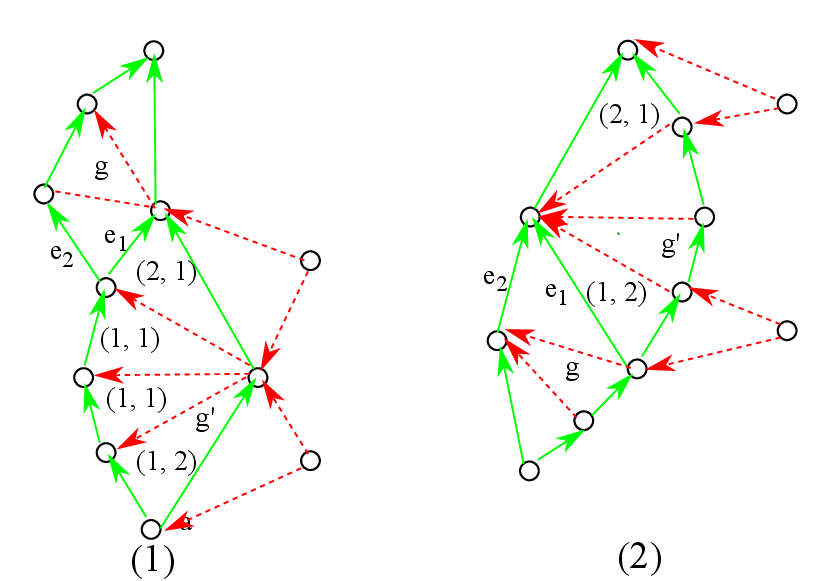}
  \centering
\caption{The proof of Lemma \ref{lemma:forbidden}: (1) $g$ is
a fan; (2) $g$ is a mirror fan.}
\label{fig:forbbiden}
\end{center}
\end{figure}

In the $\REL$ $\R$ shown in Fig \ref{fig:REL} (2), $(\FF(f,h,c),\FF(v_S,d,h))$
is a forbidden pair. So $\R$ is not a slant $\REL$.

\subsection{The Condition A3}

The following lemma specifies a necessary
and sufficient condition for adding a fan into $\tilde{G}$, and
a sufficient condition for adding a mirror fan into $\tilde{G}$.

\begin{lemma}\label{lemma:valid-fan}
Let $\R\in \PSR(G)$ and $C = C(\R)$ be its associated cut.
Let $g^L$ be a gadget and $L=\beta(g^L)$. Suppose that the conditions
A1 and A2 are satisfied for $g^L$ and $C$.
\begin{enumerate}
\item A fan $g^L$ can be added to $\R$ (i.e. $g^L$ satisfies the condition A3)
{\bf if and only if} there exists a gadget $g^R \in \R$ such that
$\beta(g^L)\subseteq \alpha(g^R)$.
\item A mirror fan $g^L$ can be added to $\R$ (i.e. $g^L$ satisfies the
condition A3) {\bf if} there exists a gadget $g^R \in \R$ such that
$\beta(g^L)\subseteq \alpha(g^R)$.
\end{enumerate}
\end{lemma}

\begin{figure}[ht]
\begin{center}
\includegraphics[width=0.95\textwidth, angle =0]{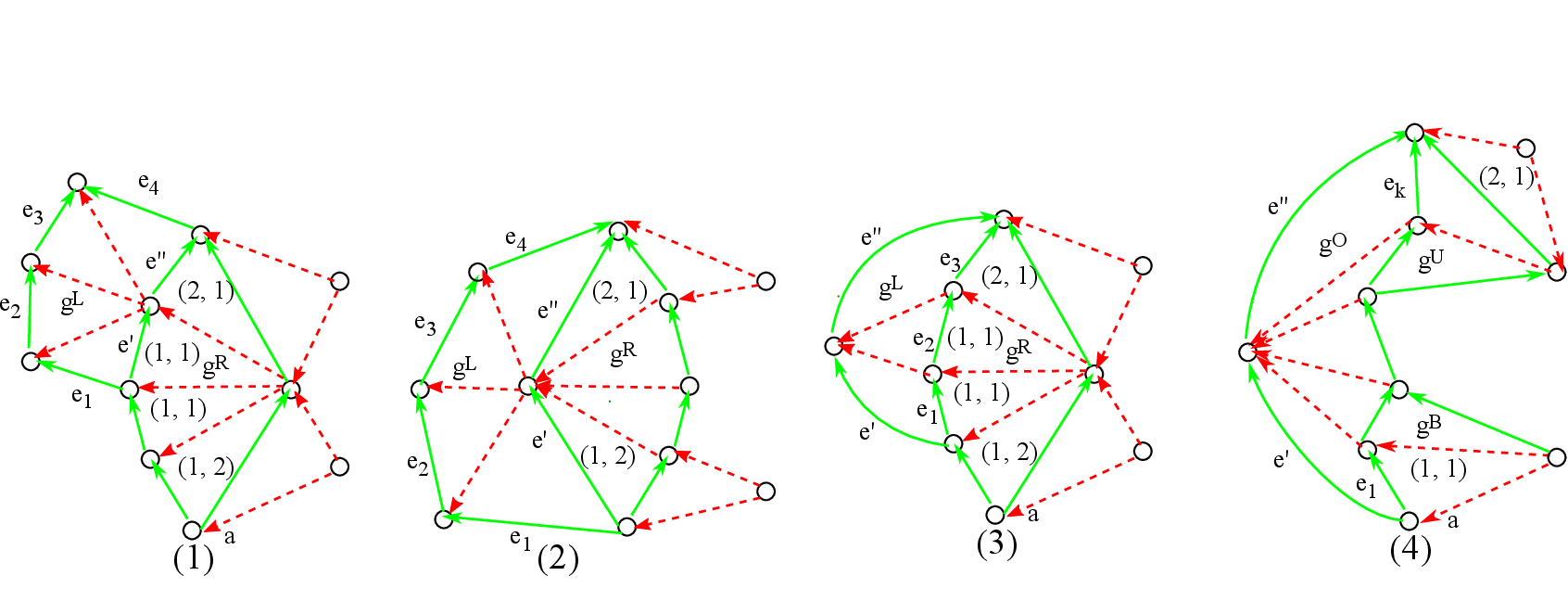}
\centering
\caption{
(1) and (2) open faces defined by edges on the front boundary of a fan $g^L$;
(3) and (4) open faces defined by edges on the front boundary of a mirror fan $g^O$.}
\label{fig:feasible}
\end{center}
\end{figure}
\vspace{-0.2in}

\begin{proof} If part of (1): Suppose there exists a gadget $g^R\in \R$
such that $\beta(g^L)\subseteq \alpha(g^R)$.
(Figs \ref{fig:feasible} (1) and (2) show two examples. In Fig
\ref{fig:feasible} (1), $g^R$ is a fan. In Fig \ref{fig:feasible} (2),
$g^R$ is a mirror fan). Let $e_1,\ldots,e_k$ be the edges in $\alpha(g^L)$.
Let $e',e''$ be the two edges in $\beta(g^L)$. Let $C'=C\otimes g^L$ be
the new cut after adding $g^L$. For each $2\leq i \leq k-1$, the type
of $\mbox{open-face}(e_i)$ is $(1,1)$.

\begin{itemize}
\item $\mbox{open-face}(e_1)$ extends $\mbox{open-face}(e')$, and
add 1 to the length of the upper side of $\mbox{open-face}(e')$.
\item $\mbox{open-face}(e_k)$ extends $\mbox{open-face}(e'')$, and
add 1 to the length of the lower side of $\mbox{open-face}(e'')$.
\end{itemize}

Regardless of where $e',e''$ are located on $\alpha(g^R)$, and regardless
of whether $g^R$ is a fan (see Fig \ref{fig:feasible} (1)) or a mirror fan
(see Fig \ref{fig:feasible} (2)), the type of $\mbox{open-face}(e_1)$
is $(1,2)$; and the type of $\mbox{open-face}(e_k)$ is $(2,1)$.
Thus $\R\otimes g^L \in \PSR(G)$.

Only if part of (1): Suppose that there exists no gadget $g^R\in \R$
such that $\beta(g^L)\subseteq \alpha(g^R)$. Let $e',e''$ be the two
edges of $\beta(g^L)$. $e'$ must be on the front boundary
of some gadget $g'$ in $\R$. $e''$ must be on the front boundary of some
gadget $g''$ in $\R$. Clearly $g' \neq g''$.
(If $g' = g''$, we would have $\beta(g^L)\subseteq \alpha(g')$).
Then either $(g^L,g')$ or $(g^L,g'')$ must be a forbidden pair.
By Lemma \ref{lemma:forbidden}, $g^L$ cannot be added to $\R$.

(2) Let $g^L$ be a mirror fan. Suppose there exists a gadget $g^R\in \R$
such that $\beta(g^O)\subseteq \alpha(g^R)$ (see Fig \ref{fig:feasible} (3)).
Similar to the proof of the if part of (1), we can show
$\R\otimes g^L \in \PSR(G)$.
\end{proof}

By Lemma \ref{lemma:valid-fan}, the only way to add a fan $g^L$ to $\R$
is by the existence of a gadget $g^R \in \R$ such that $\beta(g^L)
\subseteq \alpha(g^R)$. For a mirror fan $g$, there is another
condition for adding $g$ to $\R$ which we discuss next.

Let $v_1=v_S,v_2, \ldots, v_{t-1},v_N$ be the vertices of $C=C(\R)$ from lower
to higher order. Let $e_1$ and $e_t$ be the first and the last edge of $C$.
Imagine we walk along $C$ from $v_S$ to $v_N$. On the right
side of $C$, we pass through a sequence of gadgets in $\R$ whose front
boundary (either a vertex or an edge) touches $C$. Let
$\mbox{support}(\R)=(g_1,g_2,\ldots,g_{k-1},g_k)$,
where $e_1 \in \alpha(g_1)$ and $e_t \in \alpha(g_k)$,
denote this gadget sequence. Note that some gadgets in
$\mbox{support}(\R)$ may appear multiple times in the sequence.
(See Fig \ref{fig:chain} (1) for an example.)

Consider a mirror fan $g^O$ to be added to $\R$. Note that $L=\beta(g^O)$
is a subsequence of $C$. Let $a$ and $b$ be the lowest and the highest vertex
of $L$. Let $e_l$ be the first edge and $e_h$ be the last edge of $L$. When walking
along $L$ from $a$ to $b$, we pass through a subsequence
of the gadgets in $\mbox{support}(\R)$ on the right of $L$.
Let $\mbox{support}(L, \R) =(g^B=g_p, g_{p+1},\ldots, g_{q-1},g_q=g^U)$
denote this gadget subsequence, where:

\begin{itemize}
\item $g^B$ is the gadget such that $e_l \in \alpha(g^B)$.
\item $g^U$ is the gadget such that $e_h \in \alpha(g^U)$.
\end{itemize}

In Fig \ref{fig:chain} (1), if we add a mirror fan $g_{3}$ with
$L=\beta(g_{3})=(a, b, f, h, k, d)$, then $\mbox{support}(L,\R)
=(g_1, g_0, g_2)$.

\begin{lemma}\label{lemma:mirror-fan}
Let $\R\in \PSR(G)$ and $C = C(\R)$ be its associated cut.
Let $g^O$ be a mirror fan and $L=\beta(g^O)$. Suppose that the conditions
A1 and A2 are satisfied for $g^O$ and $C$. Let
$\mbox{support}(L, \R)=(g^B, g_{p+1}, \cdots, g_{q-1}, g^U)$.
Then $g^O$ can be added to $\R$ (i.e. $g^O$ satisfies the condition A3)
{\bf if and only if} neither $(g^O, g^B)$ nor $(g^O, g^U)$ is a forbidden pair.
\end{lemma}

\begin{proof}
First suppose that $g^O$ can be added to $\R$ to form a larger partial slant $\REL$.
Then, by Lemma \ref{lemma:forbidden}, neither $(g^O, g^B)$ nor $(g^O, g^U)$
is a forbidden pair.
Conversely, suppose that neither $(g^O, g^B)$ nor $(g^O, g^U)$
is a forbidden pair. Let $e_1,\ldots, e_k$ be the edges of $L$.
The type of $\mbox{open-face}(e_1)$ is either $(1, 1)$ or $(1, 2)$.
The type of $\mbox{open-face}(e_k)$ is either $(1, 1)$ or $(2, 1)$.
(Fig \ref{fig:feasible} (4) shows an example.)
Let $e', e''$ be the two edges in $\alpha(g^O)$.
After adding $g^O$ to $\R$, the types of $\mbox{open-face}(e')$ and
$\mbox{open-face}(e'')$ becomes $(1, 2)$ and $(2, 1)$, respectively.
They still keep the slant property for $G_2$ faces. Moreover, for each
edge $e_i~(2\leq i\leq k-1)$, $\mbox{open-face}(e_i)$ becomes
a valid complete $G_2$ face after adding $g^O$ to $\R$. Hence
$\R \otimes g^O$  is a partial slant $\REL$ of $G$.
\end{proof}

\subsection{Connections, Chains and Backbones}
\label{subsec:chains}

Given an $\R \in \PSR(G)$ and a gadget $g$, it is straightforward to check
if the conditions in Lemmas \ref{lemma:valid-fan} and \ref{lemma:mirror-fan}
are satisfied. However, as described before, maintaining the set $\PSR(G)$
requires exponential time. So we must find a way to check the conditions in
Lemmas \ref{lemma:valid-fan} and \ref{lemma:mirror-fan} without explicit
representation of $\R$.

Consider two $\R,\R'\in\PSR(G)$ such that $\R \neq \R'$ but
$\mbox{support}(\R) =\mbox{support}(\R')$.
Clearly this implies $C(\R) = C(\R')$.
By Lemmas \ref{lemma:valid-fan} and \ref{lemma:mirror-fan},
a gadget $g$ can be added to $\R$ if and only if $g$ can be added to $\R'$.
Thus, whether $g$ can be added to an $\R\in \PSR(G)$ is completely determined
by the structure of gadgets in $\mbox{support}(\R)$. There may be exponentially
many $\R' \in \PSR(G)$ with $\mbox{support}(\R') =\mbox{support}(\R)$. Instead of
keeping information of all these $\R'$, we only need to keep the information
of the structure of $\mbox{support}(\R)$. This is the main idea for converting
Algorithm \ref{alg:face-addition} to a polynomial time algorithm.
In order to describe the structure of $\mbox{support}(\R)$,
we need the following terms and notations.

\begin{definition}\label{def:connection}
Let $\R \in \PSR(G)$ and $g$ be a gadget with $L=\beta(g)$.
\begin{itemize}
\item If $\mbox{support}(L, \R)$ contains only one gadget $g^R$, and the
conditions A1, A2 and A3 are satisfied, then $(g, g^R)$ is called a
{\em $\GG$-pair}. We use $(g^L, g^R)$ to denote a $\GG$-pair.

\item If $\mbox{support}(L,\R)$ contains at least two gadgets, and the conditions
A1, A2 and A3 are satisfied, then $(g^B, g, g^U)$ is called a {\em $\MM$-triple}.
We use $(g^B, g^O, g^U)$ to denote a {\em $\MM$-triple}.

\item A $\GG$-pair $(g^L,g^R)$ or a $\MM$-triple $(g^B,g^O,g^U)$ is
called a {\em connection} and denoted by $\Lambda$.

\item For a connection $\Lambda=(g^L, g^R)$,
where $g^L=g(v^L_l, v^L_m, v^L_h), g^R=g(v^R_l, v^R_m, v^R_h)$,
the {\em front boundary} of $\Lambda$, denoted by $\alpha(g^L, g^R)$
or $\alpha(\Lambda)$, is the concatenation of the paths
$\alpha(g^R)(v^R_l, v^L_l)$, $\alpha(g^L)$, $\alpha(g^R)(v^L_h, v^R_h)$.

\item For a connection $\Lambda=(g^B, g^O, g^U)$, where
$g^B=g(v^B_l, v^B_m, v^B_h),~g^O=g(v^O_l, v^O_m, v^O_h)$,
$g^U=g(v^U_l, v^U_m, v^U_h)$, the {\em front boundary} of $\Lambda$, denoted by
$\alpha(g^B, g^O,g^U)$ or $\alpha(\Lambda)$, is the concatenation of the paths
$\alpha(g^B)(v^B_l, v^O_l), \alpha(g^O), \alpha(g^U)(v^O_h, v^U_h)$.
\end{itemize}
\end{definition}


\begin{figure}[t]
\begin{center}
\includegraphics[width=0.5\textwidth, angle =0]{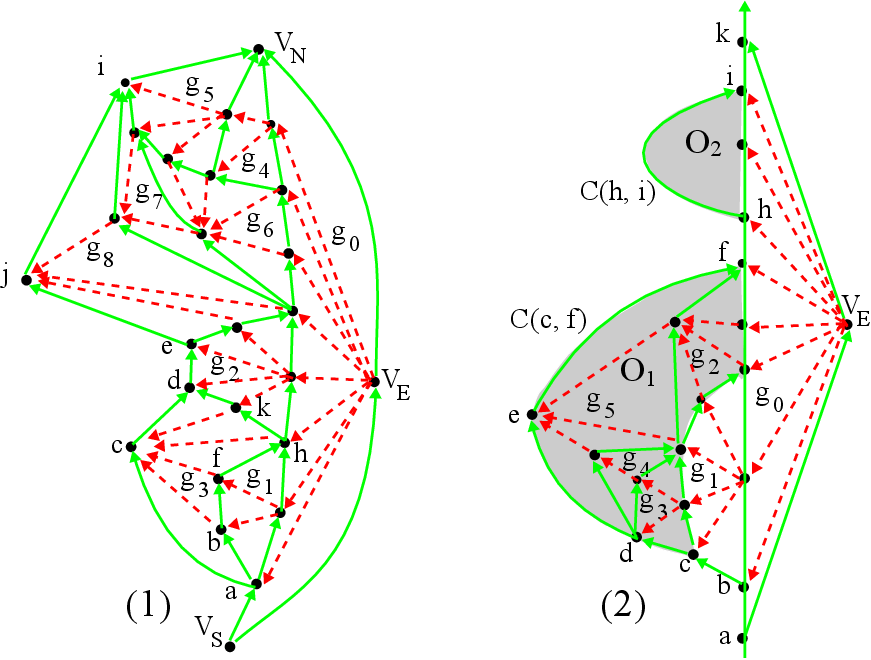}
  \centering
\caption{(1) $\R \in \PSR(G)$ is obtained by adding gadgets
$g_1, g_2, g_3, g_4, g_5, g_6, g_7,g_8$, in this order, to $g_0$.
$C(\R)=(v_S, a, c, d, e, j, i, v_N)$;
$\mbox{support}(\R) = (g_0, g_1, g_3, g_2, g_8, g_7, g_5, g_4, g_0)$ and
$(g_0 \preceq_C g_1 \preceq_C g_3 \preceq_C g_2 \preceq_C g_8 \preceq_C g_7 \preceq_C g_5 \preceq_C g_4\preceq_C g_0)$.
The pair $(g_0, g_1)$ belongs to the $\GG$-pair $\Lambda_1=(g_1,g_0)$,
the triple $(g_1, g_3, g_2)$ belongs to the $\MM$-triple
$\Lambda_2=(g_1, g_3, g_2)$, the triple $(g_2, g_8, g_7)$ belongs to the
$\MM$-triple $\Lambda_3=(g_2,g_8,g_7)$,
the pair $(g_7,g_5)$ belongs to the $\MM$-triple $\Lambda_4=(g_6,g_7,g_5)$,
the pair $(g_5, g_4)$ belongs to the $\GG$-pair $\Lambda_5=(g_5, g_4)$ and
the pair $(g_4, g_0)$ belongs to the $\GG$-pair $\Lambda_6=(g_4, g_0)$.
(2) The $\GG$-pair $\Lambda=(g_1, g_0)$ is a fractional connection
with two pockets: $\OO_1$ is bounded by $C(c,f)$ and $\alpha(\Lambda)(c,f)$
and $\OO_2$ is bounded by $C(h,i)$ and $\alpha(\Lambda)(h,i)$.
}
\label{fig:chain}
\end{center}
\vspace{-0.2in}
\end{figure}
It is tempting to think that if all gadgets in $\mbox{support}(\R)$ have
been added into $\tilde{G}$, then $\R$ has been constructed. Unfortunately,
this is not true. In order to form $\R$, the gadgets in $\mbox{support}(\R)
=(g_1, g_2,\ldots, g_{k-1}, g_k)$
must have been added to $\tilde{G}$ in the following way: When walking
along $C(\R)$ from $v_S$ to $v_N$, the gadgets in $\mbox{support}(\R)$
form a sequence $(\Lambda_1,\ldots,\Lambda_p)$
of connections such that each consecutive pair $(g_i,g_{i+1})$ or
triple $(g_{i-1},g_i,g_{i+1})$ of gadgets belong to a $\Lambda_j$
($1\leq j\leq p$); and each consecutive pair $\Lambda_i,\Lambda_{i+1}$
share a common gadget in $\mbox{support}(\R)$.
(See Fig \ref{fig:chain} (1) for an illustration).

Note that when the pair $(g_{i-1},g_i)$ and the pair $(g_i,g_{i+1})$
belong to the same connection $\Lambda_j$, it means $g_{i-1}$ and $g_{i+1}$
are the same gadget and $(g_i,g_{i+1})=(g_i,g_{i-1})$ is a $\GG$-pair.
In this case, we keep only one $\Lambda_j$ in the sequence
$\Lambda_1 \ldots, \Lambda_p$.
As seen in Fig \ref{fig:chain} (1), in addition to these connections
$\Lambda_j$ ($1 \leq j\leq p$), some gadget pairs (or triples) that are not
consecutive in $\mbox{support}(\R)$ may also form additional connections.
(In Fig \ref{fig:chain} (1), the gadgets $g_0$ and $g_2$ are not consecutive
in $\mbox{support}(\R)$. But they form a $\GG$-pair $(g_2,g_0)$). Let $\CON(\R)$
denote the set of connections formed by the gadgets in $\mbox{support}(\R)$.
(By this definition, each $\Lambda \in \CON(\R)$ has at least two
gadgets in $\mbox{support}(\R)$).
It is the structure of $\CON(\R)$ that determines if a new gadget $g$ can be
added to $\R$ or not. In general, the connections in $\CON(\R)$ cannot be
described as a simple linear structure. To describe it precisely, we need the
following definitions.

Consider a connection $\Lambda\in \CON(\R)$. If $\alpha(\Lambda) \cap C$
is a contiguous subpath of $C$, $\Lambda$ is called a {\em contiguous
connection}. If not, $\Lambda$ is called a {\em fractional connection}.
(In Fig \ref{fig:chain} (1), the $\GG$-pair $(g_1,g_0)$ is a
fractional connection.
Because the cut $C \cap \alpha(g_1, g_0)$ are
$(a, b, c)$ and $(f, h)$ and $(i, k)$, they are not a contiguous subpath of $C$.)
Consider a fractional connection
$\Lambda$. Let $u$ and $v$ be the lowest and the highest vertices
of $C\cap \alpha(\Lambda)$ respectively. When walking along $C$
from $u$ to $v$, we encounter $\alpha(\Lambda)$ multiple times.
The subpath $C(u,v)$ can be divided into a number of subpaths that are
alternatively on $\alpha(\Lambda)$, not on $\alpha(\Lambda), \ldots$,
on $\alpha(\Lambda)$. There exist at least two vertices $a,b$ in
$\alpha(\Lambda)$ such that $C(a, b)\cap \alpha(\Lambda)(a, b)=\{a, b\}$.
For each such pair of vertices $a,b$, the interior region bounded by
the subpaths $C(a, b)$ and $\alpha(\Lambda)(a, b)$ is called a {\em pocket},
denoted by $\OO=(C(a, b), \Lambda)$, of $\CON(\R)$.
Fig \ref{fig:chain} (2) shows a fractional connection $\Lambda$ (the $\GG$-pair
$(g_1, g_0)$) with two pockets $\OO_1$ and $\OO_2$.

A connection $\Lambda\in \CON(\R)$ is called
\emph{maximal} if it is not contained in any pocket of $\CON(\R)$.
A maximal connection can be either contiguous or
fractional. A non-maximal connection $\Lambda'\in \CON(\R)$ is either
completely contained in some pocket $\OO$ formed by a subpath of $C$
and a maximal fractional connection $\Lambda$ (namely all gadgets
of $\Lambda'$ are contained in $\OO$); or partially contained in $\OO$
(namely some gadget of $\Lambda'$ is contained in $\OO$
and some gadget of $\Lambda'$ is shared with $\Lambda$).
(In Fig \ref{fig:chain} (2), the $\GG$-pair $(g_3,g_1)$ and the $\MM$-triple
$(g_1,g_2,g_0)$ are partially contained in the pocket $\OO_1$. The $\MM$-triple
$(g_4,g_5,g_2)$ and the $\GG$-pair $(g_4,g_3)$ are completely contained in
$\OO_1$). Note that a pocket may contain other smaller pockets. In general,
the pockets of $\CON(\R)$ are nested in a forest-like structure.

The way to deal with fractional connections is very similar to contiguous connections.
Hence in the following paragraphs, we will assume there are no fractional connections.

\begin{definition}\label{def:lessthan}
Let two gadgets $\{g, g'\}$ belong to a connection $\Lambda$.
We say $g$ { \em precedes} $g'$ on $C$ and write
$g\preceq_C g'$ if the following conditions hold:
(1) $((\alpha(g)\cap C) \cup (\alpha(g')\cap C))$ is contiguous on $C$;
(2) When walking along $C$, we encounter the gadget $g$
before $g'$.
\end{definition}

Depending on the types of connections and their positions on a cut $C$,
there are five cases for the relation $\preceq_C$.
(They are shown in Fig \ref{fig:order}.)


\begin{figure}[t]
\begin{center}
\includegraphics[width=0.75\textwidth, angle =0]{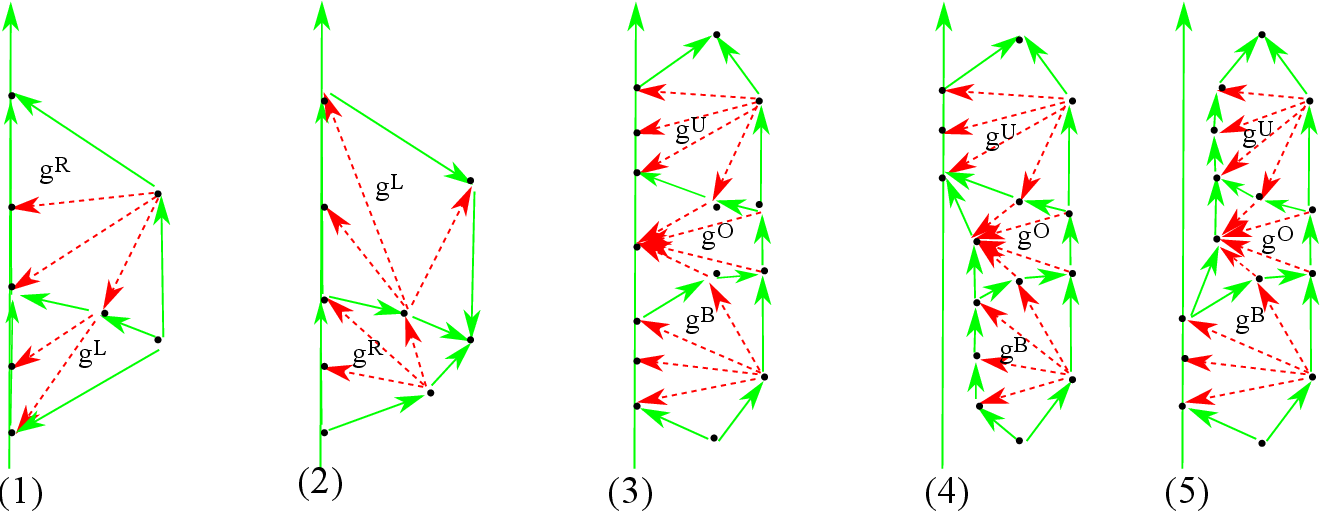}
  \centering
\caption{(1) Case 1: $(g^L, g^R)$ is a $\GG\mbox{-pair}$ and
$g_L\preceq_C g_R$;
(2) Case 2: $(g^L, g^R)$ is a $\GG\mbox{-pair}$ and
$g^R\preceq_C g^L$;
(3) Case 3: $(g^B, g^O, g^U)$ is a $\MM\mbox{-triple}$ and
$g^B \preceq_C g^O \preceq_C g^U$;
(4) Case 4: $(g^B, g^O, g^U)$ is a $\MM\mbox{-triple}$ and
$g^O \preceq_C g^U$;
(4) Case 5: $(g^B, g^O, g^U)$ is a $\MM\mbox{-triple}$ and
$g^B\preceq_C g^O$.}

\label{fig:order}
\end{center}
\end{figure}

\begin{definition}\label{def:chain}
Given a partial slant $\REL$ $\R$ with its associated cut $C$,
a sequence of gadgets $(g_1, g_2, \cdots, g_k)$ in $\mbox{support}(\R)$
is called a chain of $C$ and denoted by \emph{chain}($C$)
if the following conditions hold:
\begin{enumerate}
\item $(g_1\preceq_C g_2 \preceq_C \cdots \preceq_C g_k)$ and

\item for each $1\leq i\leq k-1$,
either $(g_i, g_{i+1})$ or $(g_i, g_{i+1}, g_{i+2})$ belongs to
a connection $\Lambda \in \CON(\R)$.
\end{enumerate}
\end{definition}

In Fig \ref{fig:chain} (1),
$(g_0 \preceq_C g_1 \preceq_C g_3 \preceq_C g_2 \preceq_C g_8 \preceq_C g_7 \preceq_C g_5 \preceq_C g_4 \preceq_C g_0)$
is a chain of $C$ where we have the $\GG$-pair $(g_1,g_0)$,
the $\MM$-triple $(g_1,g_3,g_2)$,
the $\MM$-triple $(g_2,g_8,g_7)$,
the $\MM$-triple $(g_6,g_7,g_5)$,
the $\GG$-pair $(g_5, g_4)$
and
the $\GG$-pair $(g_4, g_0)$.

Because the way a partial slant $\REL$ $\R$ is constructed,
the following property is clear.
\begin{property}\label{prop:chain}
Given a partial slant $\REL$ $\R$ with its associated cut $C$,
the $\mbox{support}(\R)=(g_1, g_2, \cdots, g_k)$
is a chain of $C$.
\end{property}

Given a partial slant $\REL$ $\R$ with its associated cut $C$,
if we can add a gadget $g$ to $\R$,
then it implies that back boundary $L=\beta(g)$ of $g$ is a part of $C$.
Let $\mbox{support}(L, \R)$ be a subsequence of $\mbox{support}(\R)$
consisting of gadgets in $\mbox{support}(\R)$ that touch $L$.
We can define an order $\preceq_L$ which is similar to $\preceq_C$.

\begin{definition}\label{def:lessthan-L}
Given a mirror fan $g^O$ with $L=\beta(g^O)$,
let two gadgets $\{g, g'\}$ belong to a connection $\Lambda$.
We say $g$ { \em precedes} $g'$ on $L$ and write
$g\preceq_L g'$ if the following conditions hold:
(1) $(\alpha(g)\cap L) \cup  (\alpha(g')\cap L)$ is contiguous on $L$;
(2) When walking along $L$, we encounter the gadgets $g$
before $g'$.
\end{definition}

\begin{definition}\label{def:backbone}
Let $g^O$ be a mirror fan with $L=\beta(g^O)$. A $\mbox{backbone}(L, \R)$
consists of a sequence of gadgets $(g_1=g^B, g_2, \cdots, g_k=g^U)$
in $\mbox{support}(L,\R)$ such that
\begin{enumerate}
\item $(g_1\preceq_L g_2 \preceq_L \cdots \preceq_L g_k)$,

\item for each $1\leq i\leq k-1$,
either $(g_i, g_{i+1})$ or $(g_i, g_{i+1}, g_{i+2})$ belongs to
a connection $\Lambda \in \CON(\R)$ and

\item Neither $(g^O, g^B)$ nor $(g^O, g^U)$ is a forbidden pair.
\end{enumerate}
\end{definition}

In Fig \ref{fig:chain} (1), consider the mirror fan $g_3$ with $L=\beta(g_3)$.
We have: $\mbox{support}(L)=(g_1, g_0, g_2)$
 where $\Lambda_1$ is the
$\GG$-pair $(g_1, g_0)$,
$\Lambda_2$ is the $\GG$-pair $(g_2, g_0)$ and
$g_1 \preceq_L g_0 \preceq_L g_2$.

Based on above discussion, we can restate Lemma \ref{lemma:mirror-fan}
as follows:

\begin{lemma}\label{lemma:valid-mirror-fan}
Let $\R\in \PSR(G)$ and $C = C(\R)$. Let $g^O$ be a mirror fan with
$L=\beta(g^O)$. Suppose that the conditions A1 and A2 are satisfied for
$g^O$ and $C$. Let $\mbox{support}(L, \R)=(g^B, g_{p+1}, \cdots, g_{q-1}, g^U)$.
Then $(g^B,g^O,g^U)$ forms a $\MM$-triple (i.e. $g^O$ satisfies the
condition A3) if and only if there exists a $\mbox{backbone}(L,\R)$
consisting of gadgets in $\mbox{support}(L,\R)$ and connections in $\CON(\R)$.
Note that each gadget and connection in $\mbox{backbone}(L, \R)$
belong to the same partial slant $\REL$ $\R$.
\end{lemma}


\section{Face-Addition Algorithm with Polynomial Time}\label{sec:impl}

We will present our polynomial time Face-Addition algorithm in this section.
In \S \ref{sec:algorithm},
we will describe the algorithm to find a superset of chains.
In \S \ref{sec:triple},
we will give more details of key procedures in \S \ref{sec:algorithm}.
In \S \ref{sec:conflict-rel},
we will present an example that Algorithm \ref{alg:efficient} may
combine two subchains of two different partial $\REL$s into a chain which only
satisfies the order $\preceq_C$ in Property \ref{prop:chain}
(there exist gadgets coming from different chains).
In \S \ref{sec:backtrack-conflict-gadgets},
we will describe a backtracking algorithm to check
whether whether a chain in the superset of chains constructed
by Algorithm \ref{alg:efficient} corresponds
to a slant $\REL$ or not.
Also, we will give runtime analysis of the backtracking algorithm.

\subsection{Polynomial Time algorithm}\label{sec:algorithm}

The polynomial time Face-Addition algorithm is described in
Algorithm \ref{alg:efficient}.

\begin{algorithm}[ht]
\caption{Face-Addition Algorithm with Polynomial Time}
\label{alg:efficient}

\KwIn {A proper triangular plane graph $G$}


Set $\tilde{\VV}=\{g_0=\FF(v_S,v_E,v_N)\}$ and $\hat{\VV} = \emptyset$\;

\Repeat{no such gadget $g$ can be found}
{
Find a gadget $g$ such that:\

\begin{description}
\item [either:] there exist v-$\GG$-pairs $(g,g^R)\notin \hat{\VV}$
with $g^R\in \tilde{\VV}$ (v-$\GG$-pairs are defined later)\;

add $g$ into $\tilde{\VV}$ (if it's not already in
$\tilde{\VV})$; add all such v-$\GG$-pairs $(g,g^R)$ into $\hat{\VV}$\;

\item[or:] $g$ is a mirror fan and there exist v-$\MM$-triples
$(g^B,g,g^U) \not\in \hat{\VV}$ with $g^B,g^U \in \tilde{\VV}$ (v-$\MM$-triples
are defined later)\;

add $g$ into $\tilde{\VV}$ (if it's not already
in $\tilde{\VV}$); add all such v-$\MM$-triples $(g^B,g,g^U)$ into $\hat{\VV}$\;
\end{description}
}

\uIf {$g_T=\MM(v_S,v_W,v_N)$ is not in $\tilde{\VV}$}
{
$G$ has no slant $\REL$\;
}
\Else{
Backtrack each v-chain of $g_T$ (v-backbone of $g_T$) to
check whether $G$ corresponds a slant $\REL$
in Algorithm \ref{alg:backtrack-conflict-gadgets};
}
\end{algorithm}

Algorithm \ref{alg:efficient} emulates the operations of Algorithm
\ref{alg:face-addition} without explicitly maintaining the set $\PSR(G)$.
Instead, it keeps two sets: (1) a set $\tilde{\VV}$ of gadgets of $G$
which contains the gadgets in the set $\tilde{G}$ defined in \S
\ref{sec:outline}, and (2) a set $\hat{\VV}$ of connections of $G$
which contains the connections in the set $\{\CON(\R)|\R\in \PSR(G)\}$ defined in \S
\ref{subsec:chains}.
In \S \ref{subsec:chains}, many concepts (cut, chain, backbone $\ldots$ etc.)
were defined referring to a $\R\in \PSR(G)$.
We now need counterparts of these concepts without referring to
a specific $\R$. For a concept {\em x} defined previously, we will use
{\em virtue x} or simply {\em v-x} for the counterpart of $x$.
(For example, v-cut for {\em virtue cut}, v-chain for {\em virtue chain},
v-backbone for {\em virtue backbone}).
A v-$\GG$-pair (v-$\MM$-triple, respectively) is similar to a $\GG$-pair
($\MM$-triple, respectively) but without referring to a specific $\R \in  \PSR(G)$.
Whenever Algorithm \ref{alg:face-addition} adds a gadget $g$ to
$\tilde{G}$ through a $\GG$-pair (or a $\MM$-triple, respectively),
Algorithm \ref{alg:efficient} adds $g$ into $\tilde{\VV}$ and add a
corresponding v-$\GG$-pair (or v-$\MM$-triple, respectively) into $\hat{\VV}$.

Initially, $\hat{\VV}$ is empty and $\tilde{\VV}$ contains only the
initial fan $g_0=\FF(v_S,v_E,v_N)$. In each step, the algorithm
finds either new v-$\GG$-pairs $(g,g^R)$ with $g^R\in \tilde{\VV}$;
or new v-$\MM$-triples $(g^B,g,g^U)$ with $g^B,g^U\in \tilde{\VV}$.
In either case, it adds $g$ into $\tilde{\VV}$.
But instead of using a $\R \in \PSR(G)$, Algorithm \ref{alg:efficient}
relies on the information stored in $\tilde{\VV}$ and $\hat{\VV}$ to
find v-$\GG$-pairs and v-$\MM$-triples.

Fix a step in Algorithm \ref{alg:efficient} and consider the sets
$\tilde{\VV}$ and $\hat{\VV}$ after this step.
Any simple path $C$ in $G$ from $v_S$ to $v_N$ is called a {\em v-cut}
of $G$. A gadget pair $(g,g^R)$ is called a {\em v-$\GG$-pair} if
$g^R \in \tilde{\VV}$ and $\beta(g) \subseteq \alpha(g^R)$.
For a v-cut $C$, define:
\begin{eqnarray*}
\tilde{\VV}(C) &=&
\{ g \in \tilde{\VV}~|~\alpha(g) \mbox{ intersects } C\}\\
\hat{\VV}(C) & = &
\{ \Lambda \in \hat{\VV}~|\mbox{ the frontiers of at least two gadgets of }\Lambda
\mbox{ intersect } C\}
\end{eqnarray*}

Let $e_1$ and $e_t$ be the first and the last edge of $C$.
A subset of gadgets $D \subseteq \tilde{\VV}(C)$
is called a {\em v-support} of $C$ if the following conditions hold:
\begin{itemize}
\item The gadgets in $D$ can be arranged into a sequence $(g_1,g_2,\ldots, g_k)$
such that
$e_1 \in \alpha(g_1)$, $e_t \in \alpha(g_k)$ and, when walking along
$C$ from $v_S$ toward $v_N$, we encounter these gadgets in this order.

\item Any two (or three) consecutive gadgets $(g_i,g_{i+1})$
(or $(g_{i-1},g_i,g_{i+1})$) belong to a connection in $\hat{\VV}$.
\end{itemize}

If a set $S$ of connections formed by the gadgets in a v-support of $C$
satisfies the structure property described in Definition \ref{def:chain},
$S$ is called a {\em v-chain} of $C$.  Clearly, any chain is also a
v-chain.

Let $g$ be a mirror fan with $L=\beta(g)$. Let $a$ and $b$ be
the lowest and the highest vertex of $L$, and $e_l$ and $e_h$
the first and the last edge of $L$, respectively. Define:
\begin{eqnarray*}
\tilde{\VV}(L) &=&
\{ g \in \tilde{\VV}~|~\alpha(g) \mbox{ intersects } L\}\\
\hat{\VV}(L) & = &
\{ \Lambda \in \hat{\VV}~|\mbox{ the frontiers of at least two gadgets of }\Lambda
\mbox{ intersect } L\}
\end{eqnarray*}

A subset of gadgets $D \subseteq \tilde{\VV}(L)$
is called a {\em v-support} of $L$ if the following conditions hold:
\begin{itemize}
\item The gadgets in $D$ can be arranged into a sequence $(g^B=g_p,g_2,\ldots,
g_q=g^U)$ such that
$e_l\in\alpha(g^B)$, $e_h \in \alpha(g^U)$ and, when walking along $L$
from $a$ toward $b$, we encounter these gadgets in this order.

\item Any two (or three) consecutive gadgets $(g_i,g_{i+1})$
(or $(g_{i-1},g_i,g_{i+1})$) belong to a connection in $\hat{\VV}$.
\end{itemize}

If a set $S$ of connections formed by the gadgets in a v-support of $L$
only satisfies the order property $\preceq_L$ and
the third property described in Definition \ref{def:backbone},
$S$ is called a {\em v-backbone} of $L$. If there is a
v-$\mbox{backbone}(L)$, we call $(g^B,g,g^O)$ a {\em v-$\MM$-triple}.
Both v-$\GG$-pairs and v-$\MM$-triples are called v-connections.

First, we bound the number of loop iterations in Algorithm \ref{alg:efficient}.
By Lemma \ref{lemma:number}, the number of gadgets in $G$ is at most
$N=O(n^2)$. So the number of v-$\GG$-pairs is at most $O(N^2)$ and the
number of v-$\MM$-triples is at most $O(N^3)$. Hence $\hat{\VV}$ contains
at most $O(n^6)$ elements. Since each iteration adds at least either
a v-$\GG$-pair or a v-$\MM$-triple into $\hat{\VV}$, the number of
iterations is bounded by $O(n^6)$.

We need to describe how to perform the operations in the loop body,
which is clearly dominated by finding v-$\GG$-pairs and
finding v-$\MM$-triples. Given two gadgets $g$, $g^R$ and the sets
$\tilde{\VV}$ and $\hat{\VV}$, it is easy to check if $(g,g^R)$
is a v-$\GG$-pair (i.e. $g^R \in \tilde{\VV}$ and $\beta(g)\subseteq \alpha(g^R)$)
in polynomial time. However,
finding v-$\MM$-triples ($g^B,g,g^U)$ is
much more difficult. In \S \ref{sec:triple}, we show this can be done,
in polynomial time, by finding a v-$\mbox{backbone}(\beta(g))$ consisting of
connections in $\hat{\VV}$. This will establish the polynomial run time of
the repeat loop of Algorithm \ref{alg:efficient}.
\begin{lemma}\label{lemma:sufficient}
Let $\SSS$ be the set of all v-backbones of $g_T$.
(Because $L=\beta(g_T)$ is a v-cut,
each v-backbone of $L$ is actually a v-chain of $G$.)
For each $\R\in \PSR(G)$ with its associated cut $C=C(\R)$,
there exists a v-chain $S\in \SSS$ (which is a v-backbone
of $L$) generated by Algorithm \ref{alg:face-addition}
such that $S=\mbox{chain}(C)$.
\end{lemma}

\begin{proof}
For each mirror fan $g$ with $L=\beta(g)$,
if $g$ is in some partial slant $\REL$,
then its backbone follows the order $\preceq_L$.
So we have $\tilde{G} \subseteq \tilde{\VV}$ and
$\{\CON(\R)|\R\in \PSR(G)\} \subseteq \hat{\VV}$.
Since we use the two supersets $\tilde{\VV}$ and $\hat{\VV}$
of $\tilde{G}$ and $\hat{G}$ to find v-backbones of $g_T$ (chains of $G$),
we immediately have:

$\{\mbox{chain}(C)|\R \mbox{ is a slant } \REL
\mbox{ of } G \mbox{ with } \mbox{its } \mbox{associated } \mbox{cut }
C=C(\R)\} \subseteq \SSS.$
\end{proof}

In this subsection, we have described Algorithm \ref{alg:efficient}
which constructs
(1) the set $\SSS$ of v-chains such that
chain($C$) of each partial $\REL$ $\R\in \PSR(G)$ with its
associated cut $C=C(\R)$ is included in $\SSS$,
(2) the set $\hat{\VV}$ of v-connections containing
each connection $\Lambda \in \{\CON(\R)|\R\in \PSR(G)\}$,
(3) the set $\tilde{\VV}$ of gadgets containing
each gadget $g\in \tilde{G}$.

Lemma \ref{lemma:sufficient} states that any chain is a v-chain.
But the reverse is not necessarily true.
In \S \ref{sec:conflict-rel}, we provide
an example that a v-chain
is not equal to the chain of
the associated cut of any partial slant $\REL$.
Thus we need to check whether a v-chain constructed by
Algorithm \ref{alg:efficient} is really
a chain of a slant $\REL$ $\R$ of $G$.
In \S \ref{sec:backtrack-conflict-gadgets},
we describe a backtracking algorithm to detect all such v-chains.

\subsection{Algorithm for Finding v-Backbones and v-$\MM$-triples}
\label{sec:triple}

Consider a gadgets $(g^B, g^O, g^U)$ with the back boundary $L=\beta(g^O)$.
Let $a$ and $b$ be the lowest and the highest vertex of $L$,
$e_1$ and $e_2$ the first
and the last edge of $L$. In this subsection,
we show how to check whether $(g^B, g^O, g^U)$ is a v-$\MM$-triple or not
in polynomial time.
By the definition of v-$\MM$-triples,
this is equivalent to finding
$\mbox{v-backbone}(L)$s by using the v-connections in $\hat{\VV}$.

Let $\tilde{\VV}(L)$ be the set of gadgets in $\tilde{\VV}$ that can be in
any $\mbox{v-backbone}(L)$. From the conditions described in
Definition \ref{def:backbone}, $\tilde{\VV}(L)$ contains the gadgets
$g \in \tilde{\VV}$ that satisfy the following conditions:

\begin{itemize}
\item The front boundary of $g$ intersects $L$ and $g$ belongs to
some v-connection $\Lambda \in \hat{\VV}$.

\item the front boundary of $g^B$ contains $e_1$;
and $(g^O,g^B)$ is not a forbidden pair.

\item the front boundary of $g^U$ contains $e_2$;
and;
and $(g^O,g^U)$ is not a forbidden pair.
\end{itemize}

To determine which gadgets in $\tilde{\VV}(L)$ can form
a v-backbone of $L$, we construct a directed acyclic graph
$\HH_L=(V_L,E_L)$  as follows:

\begin{definition}\label{gra:L}
Given a triple $(g^B, g^O, g^U)$ with $L=\beta(g^O)$,
the \emph{backbone graph} $\HH_L$ of $g^O$ is defined as follows:
\begin{itemize}

\item $V_L=\{(l, g)| l=L\cap \alpha(g) \mbox{ and }
g\in \tilde{\VV}(L) \}$ and

\item $E_L=\{(l_1, g_1)\rightarrow (l_2, g_2)| $

$\{g_1, g_2\}$
belongs to a v-connection $\Lambda\in \hat{\VV}(L)$ and
$l_1\cup l_2$ is contiguous on $L$ $\}$

\end{itemize}
\end{definition}
A \emph{source} (\emph{sink}, respectively) vertex
has no incoming (outgoing, respectively) edges in $\HH_L$.
The intuitive meaning of a directed path $P\in \HH_{L=\beta(g^O)}$
from the source to the sink
is that $P$ corresponds to a v-backbone of $L=\beta(g^O)$ and
for each vertex $(l, g) \in P$, $g$ corresponds a gadget in a v-backbone
and $l$ is equal to the intersection $L\cap \alpha(g)$ of $L$
and the front boundary of $g$.
Moreover, for different v-$\MM$-triples of $g^O$
$\Lambda_1=(g^{B_1}, g^O, g^{U_1}),
\Lambda_2=(g^{B_2}, g^O, g^{U_2})\in \hat{\VV}(L)$,
we have vertices $v_O=(L\cap \alpha(g^O), g^O)$ and $v'_O=(L\cap \alpha(g^O), g^O)$
in $\HH_L$ to represent $g^O$
such that $\Lambda_1$ and $\Lambda_2$ represent different subpaths in $\HH_L$:
one is
$(L\cap \alpha(g^{B_1}), g^{B_1})
\rightarrow v_O=(L\cap \alpha(g^O), g^O) \rightarrow (L\cap \alpha(g^{U_1}), g^{U_1})$
and the other one is
$(L\cap \alpha(g^{B_2}), g^{B_2})
\rightarrow v'_O=(L\cap \alpha(g^O), g^O) \rightarrow (L\cap \alpha(g^{U_2}), g^{U_2})$
where the vertex $v_O$ represents the mirror fan in $\Lambda_1$
and $v'_O$ represents the mirror fan in $\Lambda_2$.

\begin{lemma}\label{lemma:GL}
$\HH_L$ is acyclic and can be constructed in $O(|\hat{\VV}|^2)$ time.
\end{lemma}

\begin{proof}
Consider a $g \in \tilde{\VV}$. Knowing $L$, we can easily determine
if $g$ is in $\tilde{\VV}(L)$ in constant time. So we can identify
the set $V_L$ in $O(|\hat{\VV}|)$ time.
For two vertices $(l, g)$ and $(l', g')$ in $V_L$, the edge
$(l, g) \rightarrow (l', g')$ exists if and only if the following
two conditions are satisfied:
(1) $g$ and $g'$ belong to some v-connection in $\hat{\VV}(L)$;
(2) $l\cup l'$ is contiguous on $L$; and
(2) when walking along $L$ upwards, we encounter the gadgets $g$
before $g'$.
These two conditions can be easily checked in constant time.
So the set $E_L$ can be determined
in $O(|V_L|^2)=O(|\hat{\VV}|^2)$ time. Thus $\HH_L$ can be constructed
in $O(|\hat{\VV}|^2)$ time.

The edge directions of $\HH_L$ are defined by the relation $\preceq_L$.
Since $\preceq_L$ is acyclic, $\HH_L$ is acyclic.
\end{proof}


\begin{lemma}\label{lemma:feasibility}
Given a triple $(g^B, g^O, g^U)$ with $L=\beta(g^O)$, let $\HH_L$ be the graph
defined in Definition \ref{gra:L}.

\begin{enumerate}
\item Each directed path from the source to the sink
in $\HH_L$ corresponds to the v-$\MM$-triple $(g^B, g^O, g^U)$.



\item The v-$\MM$-triple $(g^B, g^O, g^U)$ corresponds to a set of
directed paths from $g^B$ to $g^U$ in $\HH_L$.

\end{enumerate}
\end{lemma}

\begin{proof}
Statement 1.
Consider any directed path $(l^B, g^B) \rightarrow \cdots
\rightarrow (l^U, g^U)$ from the source $(l^B, g^B)$ to the sink
$(l^U, g^U)$ in $\HH_L$.
Since each directed edge $(l, g)\rightarrow (l', g')$ in $\HH_L$
follows the order $\preceq_L$ on $L$,
each directed path from $(l^B, g^B)$ to $(l^U, g^U)$ is a v-backbone of $g^O$
and $(g^B,g^O,g^U)$ is a v-$\MM$-triple.

Statement 2: Consider a v-$\MM$-triple $(g^B,g^O,g^U)$. This means
that there exists a v-backbone
$g^B\preceq_L g_2 \preceq_L \cdots \preceq_L g_{k-1}\preceq_L g^U$ on $L$.
Because each $g\preceq_L g'$ on $L$ is a directed edge $(l, g)\rightarrow (l', g')$
in $\HH_L$,
we have that
$(l^B, g^B)\rightarrow (l_2, g_2) \rightarrow \cdots \rightarrow (l_{k-1}, g_{k-1}) \rightarrow (l^U, g^U)$
is a directed path from $(l^B, g^B)$ to $(l^U, g^U)$ in $\HH_L$.

Note that there may exist multiple paths in $\HH_L$ from $(l^B, g^B)$ to
$(l^U, g^U)$. All these paths correspond to the same v-$\MM$-triple
$(g^B,g^O,g^U)$. The intuitive meaning of this fact is as follows.
When we add $g^O$ via the v-$\MM$-triple $(g^B,g^O,g^U)$, even though
the gadgets $g^B$ and $g^U$ are fixed, the v-connections and gadgets in the
$\mbox{v-backbone}(L)$s between $g^B$ and $g^U$
may be different. But as long as they form a valid
$\mbox{v-backbone}(L)$, we can add $g^O$.
\end{proof}

The following Algorithm \ref{alg:triple}  finds
v-$\MM$-triples $(g^B, g^O, g^U)$ by finding $\mbox{v-backbone}(L)$s.

\begin{algorithm}[htb]
\caption{~Find v-$\MM$-triples}
\label{alg:triple}

\KwIn {A triple $(g^B, g^O, g^U)$ with $L=\beta(g^O)$ and the set $\hat{\VV}$ of
v-connections}

From the connections in $\hat{\VV}$, identify the set $\hat{\VV}(L)$\;

Construct the directed graph $\HH_L$ as in Definition \ref{gra:L}\;

By using Lemma \ref{lemma:feasibility},
return whether $(g^B, g^O, g^U)$ is $v$-$\MM$-triple or not\;
\end{algorithm}

\begin{theorem}\label{thm:time}
Given a gadgets triple $(g^B, g^O, g^U)$,
Algorithm \ref{alg:triple} can successfully test whether $(g^B, g^O, g^U)$
is a v-$\MM$-triple  in polynomial time.
\end{theorem}

\begin{proof}
The correctness of the algorithm follows from Lemma
\ref{lemma:feasibility}. By Lemma \ref{lemma:GL}, the steps
1 and 2 can be done in polynomial time.

Step 3:
Since $\HH_L$ is acyclic,
we can use breadth-first search to find whether
$(l^U, g^U)$ is reachable from $(l^B, g^B)$.
Then $(g^B, g, g^U)$ is
a v-$\MM$-triple if and only if
$(l^U, g^U)$ is reachable from $(l^B, g^B)$.
This step is carried out by calling breadth-first search which
takes polynomial time.
So the total time for this step is polynomial.

Note that the total number of source to sink paths in $\HH_L$ can be exponential.
However, we only
need to find one path from $(l^B, g^B)$ to $(l^U, g^U)$.

\end{proof}

\subsection{An Example that A v-Chain Does Not Have
A Slant $\REL$}\label{sec:conflict-rel}
In this subsection, we present an example to show
why a v-chain defined in the last subsection does not necessarily
have a corresponding partial slant $\REL$.
Imagine that we have two v-chains $\mathcal{C}$ and $\mathcal{C}'$.
Suppose that $\mathcal{C}$ can be partitioned into
$\mathcal{C}=(\mathcal{C}_1, \mathcal{C}_2, \mathcal{C}_3)$ and
$\mathcal{C}'$ can be partitioned into
$\mathcal{C}'=(\mathcal{C}'_1, \mathcal{C}'_2, \mathcal{C}'_3)$
such that $\mathcal{C}_2=\mathcal{C}'_2$,
then we may have another two v-chains
$(\mathcal{C}_1, \mathcal{C}_2=\mathcal{C}'_2, \mathcal{C}'_3)$
and $(\mathcal{C}'_1, \mathcal{C}_2=\mathcal{C}'_2, \mathcal{C}_3)$.
However, both of the two v-chains can't correspond any partial slant $\REL$.
Fig \ref{fig:conflict} (3) shows a v-chain $(g_1, g_2, g_3, g_4, g_5, g_6)$
which can't have a corresponding partial slant $\REL$ where
$(g_1, g_2, g_3)$ from $\R_1$ and $(g_3, g_4, g_5, g_6)$ from $\R_2$ shares a common gadget $g_3$.




\begin{figure}[t]
\centering
\includegraphics[width=0.85\textwidth, angle =0]{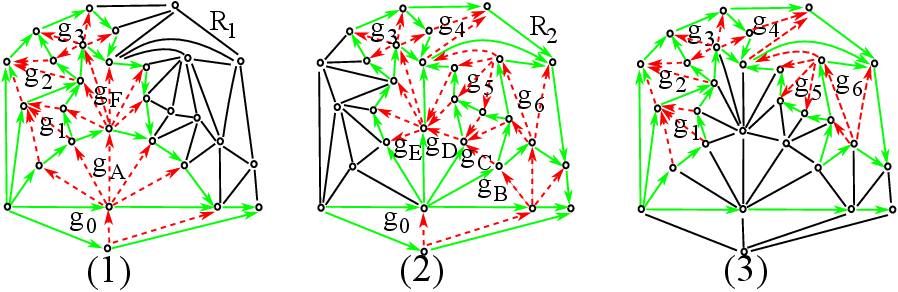}
\caption{(1) is a partial slant $\REL$ $\R_1 $ which
consists of gadgets $\{g_0, g_A, g_F, g_1, g_2, g_3\}$;
(2) is a partial slant $\REL$ $\R_2$ which
consists of gadgets $\{g_0, g_3, g_4, g_5, g_6, g_B, g_C, g_D, g_E\}$;
(3) $(g_1, g_2, g_3, g_4, g_5, g_6)$ is a v-chain where
$(g_1, g_2, g_3)$ is a subchain of $\R_1$ and
$(g_3, g_4, g_5, g_6)$ is a subchain of $\R_2$.
However, the v-chain is not coming from the same partial slant $\REL$.
$\{g_1, g_2, g_3\}$ can be added into $\R_1$ only when $\{g_0, g_A\}$ have been
added into $\R_1$.
The order of added gadgets in $\R_2$ is:
$(g_0, g_B, g_6, g_5, g_C, g_D, g_E, g_3, g_4)$.
But, $g_A$ and each gadget of $\{g_B, g_C, g_D, g_E\}$ can not coexist in the
same $\REL$ because
some faces of $g_A$ and each gadget of $\{g_B, g_C, g_D, g_E\}$
overlap.
}\label{fig:conflict}
\end{figure}

\subsection{An Algorithm to Find Conflicting Gadgets via Backtracking}\label{sec:backtrack-conflict-gadgets}

In the last subsection,
we know that each v-chain of $g_T$ only contains partial information of a
complete slant $\REL \R$.
In this subsection, we use a recursive constructive definition
to define a \emph{hierarchal v-chain}
which represents
sufficient information of a complete $\REL$ $\R$ and
can be represented by a DAG
as follows:
\begin{definition}\label{def:h-v-chain}
Given the final mirror fan $g_T$ with $C=\beta(g_T)$,
a \emph{hierarchal v-chain} $\JJ=(V(\JJ), E(\JJ))$ of $C$ is a DAG
recursively defined as follows:
\begin{enumerate}
\item The \emph{root} $\JJ(r)\in V(\JJ)$
is a sequence of pairs $((C_1, g_1), (C_2, g_2), \cdots, (C_k, g_k))$
where
\begin{enumerate}
\item $(g_1, g_2, \cdots, g_k)$ is a v-chain
$(g_1\preceq_C g_2 \preceq_C \cdots \preceq_C g_k)$ of $C$ and

\item each $C_i=C\cap \alpha(g_i), 1\leq i\leq k,$ is a portion of the front boundary $\alpha(g_i)$ of $g_i$.
\end{enumerate}

\item $\mbox{While}(\alpha(g_0) \nsubseteq C)$

\begin{enumerate}
\item select a gadget $g$ such that
\begin{enumerate}
\item $\alpha(g) \subseteq C$ and
\item there exist a sequence of pairs
$((l_1, g)\in S_1, (l_2, g)\in S_2, \cdots, (l_h, g)\in S_h)$
where each $S_i, 1\leq i\leq h,$ is a vertex of $\JJ$ and
$l_1\cup l_2 \cup \cdots \cup l_h = \alpha(g)$,
\end{enumerate}


\item create a vertex $S$ consisting of a sequence of pairs
$((l_1, g_1), (l_2, g_2), \cdots, (l_h, g_h))$
where
\begin{enumerate}
\item $(g_1\preceq_L g_2\preceq_L \cdots \preceq_L g_h)$
is a v-backbone of $L=\beta(g)$ and
\item each $l_i=L\cap \alpha(g_i), 1\leq i\leq h,$ is a portion
of the front boundary $\alpha(g)$ of $g$,
\end{enumerate}

\item add $S$ into $V(\JJ)$ and
for each $S_i, 1\leq i\leq h$,
add an arc $S_i \rightarrow S$ into $E(\JJ)$. And,

\item change $C$ to $C(v_S, a) \cup \beta(g) \cup C(b, v_N)$
where $a$ and $b$ are the first and the last vertices of $\beta(g)$, respectively.

\end{enumerate}

\end{enumerate}

\end{definition}
Intuitively a hierarchal v-chain $\JJ$ is a hierarchal decomposition of
a complete slant $\REL$ $\R$ and
the root $\JJ(r)$ of $\JJ$ represents a chain($C$) of $\R$'s associated cut $C=C(\R)$.
In the following definition,
a hierarchial structure $\HH$ consists of a set of DAGs (backbone graphs)
and $\HH$ can implicitly store all possible hierarchal v-chains $\JJ$.

\begin{definition}
$\HH=(V(\HH), E(\HH))$ is a DAG where

\begin{enumerate}

\item for each vertex $v\in V(\HH)$, $v$ represents a DAG
$\HH(v)=(V(\HH(v)), E(\HH(v)))$ over $V(\HH(v))$ where
\begin{enumerate}
\item every vertex $w \in V(\HH(v))$
is a pair $(l, g)$ and $l$ is a portion
of the front boundary $\alpha(g)$ of $g$,

\item for each arc $e=(u_e\rightarrow u'_e)$ in $E(\HH(v))$, $e$ associates with a DAG $\HH(v'), v'\in V(\HH)$
(the associated DAG of $e$ is denoted by $\HH(e)$),
the source of $\HH(e)$ is the starting vertex $u_e$ of $e$ and
the sink of $\HH(e)$ is the ending vertex $u'_e$ of $e$.
\end{enumerate}

\item an ordered pair $(u, v)$ belongs to $E(\HH)$ if there exists an arc $e$ in $E(\HH(u))$ such that
$e$'s associated DAG $\HH(e)$ is equal to $\HH(v)$.
\end{enumerate}
We call an ordered pair vertices $(u, v) \in E(\HH)$ a $\emph{super arc}$ of $\HH$.
Also, for each arc $e\in E(\HH(u))$,
let $\HH_e$ be the maximal subgraph of $\HH$ which
can be reached from $\HH(e)$ via super arcs.
\end{definition}



From now on,
(1) when we mention a DAG $\HH(e)$ from an arc, it means that
the arc $e$ is in the DAG represented by a vertex in $V(\HH)$,
(2) when we mention a DAG $\HH(v)$ represented by a vertex $v$,
it means that the vertex $v$ is a vertex in $V(\HH)$, and
(3) we use the term $\HH(e_C)$ to represent the DAG
in the root of the hierarchal structure $\HH$.

Given a fan $g$ and an arc $e=(l_1, g_1)\rightarrow (l_2, g_2)\in E(\HH(v)),$
(1) $e$ is a \emph{complete arc} on $g$ if
$g=g_1=g_2$,
(2) $e$ is a \emph{left partial arc} on $g$ if
$g\neq g_1$ and $g=g_2$,
(3) $e$ is a \emph{right partial arc} on $g$ if
$g=g_1$ and $g\neq g_2$, and
(4) $e$ is \emph{minimal} if
$l_1\cup l_2$ is a contiguous path.
Algorithm \ref{alg:backtrack-conflict-gadgets}
emulates the growing process of all hierarchal v-chains $\JJ$ as follows:
\begin{enumerate}
\item add the root $v$ into $\HH$ and
let $\HH(r)$ be the backbone graph $\HH_{C=\beta(g_T)}$
of $g_T$.
Now, $V(\HH)=\{r\}$ and
the root $\JJ(r)$ of each hierarchal v-chain $\JJ$ is a directed path $P\in \HH(v)$,
and vice versa.

\item iteratively selects a gadget $g$
(to be defined in Definition \ref{def:removable-gadget})
such that
\begin{description}
\item [if] $g$ is a mirror fan and has a path
$P=(\cdots, (l^B, g^B), (l=\alpha(g), g), (l^U, g^U), \cdots) \in \HH(v)$,
(1) add a vertex $v'$ into $V(\HH)$,
(2) let $\HH(v')$ be the backbone graph $\HH_{L'=\beta(g)}$ of $g$,
(3) change $P$ to $(\cdots, (l^B, g^B), (l^U, g^U), \cdots) \in \HH(v)$,
(4) add a super arc from $v$ to $v'$ in $E(\HH)$ and
(5) let $\HH(e')$ be $\HH(v')$ where $e'=(l^B, g^B) \rightarrow (l^U, g^U)$.
The backbone graph of $g$ is embedded into $\HH(e')$.
See Fig \ref{fig:expand} as an example.


\item [Otherwise,] $g$ is a fan.
For each maximal path $P=(v_1, v_2, \cdots, v_k)\in \HH(v), v\in V(\HH)$
where the gadget $g_i$ of each $v_i=(l_i, g_i)$ is equal to $g$,
(1) merge $P$,
(2) add an arc $e'$ between $v_1$ and $v_k$
and (3) set $\HH(e')=(\beta(g)\cap \alpha(g^R), g^R)$
(the backbone graph of $g$)
where $(g, g^R)$ is a v-$\GG$-pair in $\hat{\VV}$.
The backbone graph of $g$ is embedded into $\HH(e')$.
Figs \ref{fig:merge} (1) and (2) show an example of $P$
before merging $P$ and
Figs \ref{fig:merge} (3) and (4) show an example of $P$
after merging $P$.
\end{description}
\end{enumerate}

\begin{figure}[t]
\centering
\includegraphics[width=0.75\textwidth, angle =0]{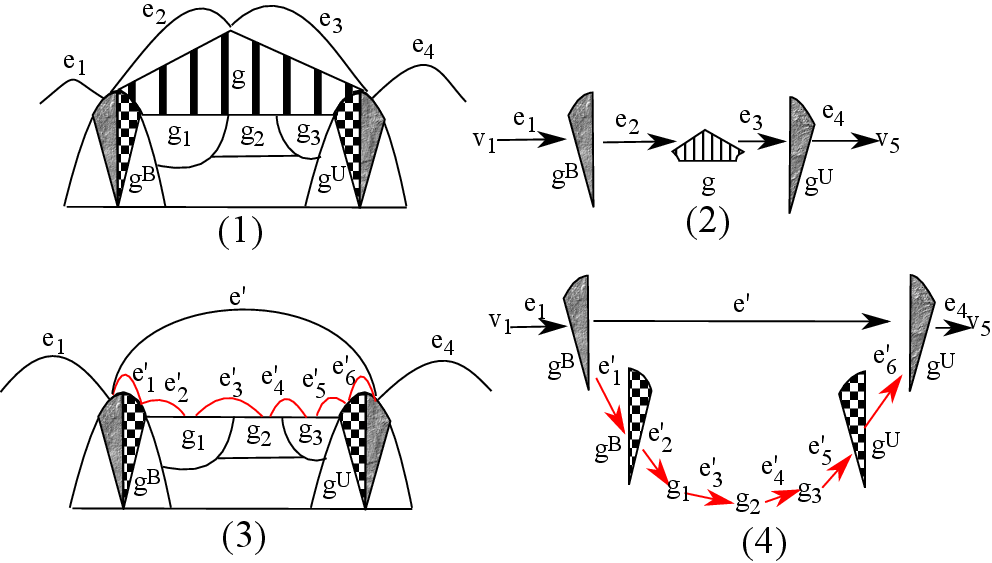}
\caption{(1) and (2) show a $\MM$-triple $(g^B, g, g^U)$ and
suppose that $\HH_C$ has a directed path $P=(\cdots,$ $e_1, e_2, e_3, e_4,$ $\cdots)$;
(3) and (4) show that after removing $g$, we add a new arc $e'$ into
$\HH_{e_C}$ and
$P$ becomes $(\cdots,$ $e_1, e', e_4,$ $\cdots)\in \HH_C$. And,
$(e'_1, e'_2,$ $e'_3, e'_4,$ $e'_5, e'_6)$ is a directed path in $\HH(e')$
where $(e'_1, e'_2,$ $e'_3, e'_4,$ $e'_5, e'_6)$
is a v-backbone of $L=\beta(g)$.
}
\label{fig:expand}
\end{figure}

\begin{figure}[t]
\centering
\includegraphics[scale=0.75, angle =0]{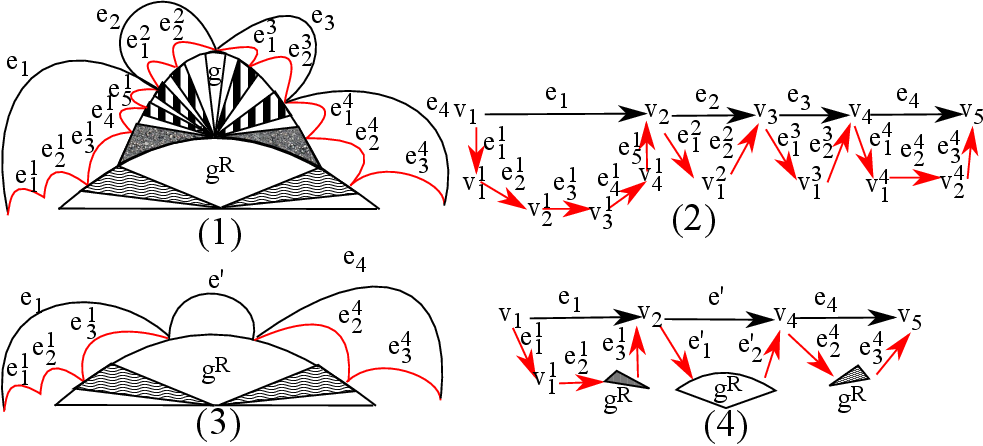}
\caption{
(1) and (2) are an example of $\HH$, complete arcs and partial arcs;
(3) and (4) are an example to explain how $\HH$ changes its structure after removing
a fan $g$;
(1) and (2): $(\cdots, e_1, e_2, e_3, e_4, \cdots)$ is a directed path $P\in \HH(e_C)$
where $e_1$ is a left partial arc on $(g, g^R)$,
$\{e_2, e_3\}$ are complete arcs on $(g, g^R)$ and
$e_4$ is a right partial arc on $(g, g^R)$.
Also, $(e^1_1, e^1_2,$ $e^1_3, e^1_4,$ $e^1_5)$
is a directed path in $\HH(e_1)$,
$(e^2_1, e^2_2)$ is a directed path in $\HH(e_2)$,
$(e^3_1, e^3_2)$ is a directed path in $\HH(e_3)$ and
$(e^4_1, e^4_2, e^4_3)$ is a directed path in $\HH(e_4)$;
(3) and (4): after removing the fan $g$,
change the vertices $v_2=(l_2, g)$ and $v_4=(l_4, g)$ to
$v_2=(\beta(g)\cap \alpha(g^R), g^R)$ and $v_4=(\beta(g)\cap \alpha(g^R), g^R)$, respectively
where $\beta(g) \cap \alpha(g^R)$ is the intersection of the back boundary $\beta(g)$
and the front boundary $\alpha(g^R)$,
and $(g, g^R)$ is a $\GG$-pair in $\tilde{\VV}$.
Also, the arcs $\{e_2, e_3\}$ are replaced by the arc $e'$
and $\HH(e')$ is the path $(v_2, (\beta(g)\cap \alpha(g^R), g^R), v_4)$.
}
\label{fig:merge}
\end{figure}

The next definition
defines a \emph{removable} gadget $g$
which can be selected in Algorithm \ref{alg:backtrack-conflict-gadgets}
and add the backbone graph $\HH_{L=\beta(g)}$ into $\HH$ .
Intuitively a \emph{removable} gadget $g$ means that
all gadgets $g'$ which are connections $(g', g)$ in $\hat{\VV}$
have been selected and removed from Algorithm \ref{alg:backtrack-conflict-gadgets}.

\begin{definition}\label{def:removable-gadget}
In Algorithm \ref{alg:backtrack-conflict-gadgets},
we say a gadget $g$ is \emph{removable} from a DAG $\HH(e)$ $V(\HH)$
if there exists a vertex $(l, g) \in \HH(e)$ and
we cannot find a vertex $(l', g')$ from another DAG $\HH(e')$ in $V(\HH)$
such that
$g$ and $g'$ belong to some connection $\Lambda\in \hat{\VV}$.
Note that the vertex $(l', g')$ can also be selected from $\HH(e)$.
\end{definition}

Now we give the definition of a \emph{conflicting} hierarchal v-chain $\JJ$ which
cannot form a slant $\REL$ $\R$. An example for
a \emph{conflicting} hierarchal v-chain $\JJ$ has been shown in Fig \ref{fig:conflict}.
\begin{definition}\label{def:non-conflicting-v-chain}
A hierarchal v-chain $\JJ$
is \emph{conflicting} on a gadget $g$ if
there exist pairs $(l, g)\in S$ and $(l', g')\in S'$ where
$S$ and $S'$ are two vertices in $V(\JJ)$
such that
\begin{enumerate}
\item if $g \neq g'$, $g$ and $g'$ overlap at least one face.
\item Otherwise ($g = g'$), $l$ and $l'$ overlap at least two vertices.
\end{enumerate}
Note that $S$ might be equal to $S'$.
Moreover, we say
the vertex $(l', g')$ is \emph{conflicting} to $(l, g)$ on $g$
if $(l, g)$ and $(l', g')$ satisfy one of the above two conditions.
On the other hand,
we say $(l', g')$ is \emph{compatible} to $(l, g)$ on $g$
if $(l', g')$ is not conflicting to $(l, g)$ on $g$.
And,
for a hierarchal v-chain $\JJ$, we say $\JJ$ is \emph{compatible} on $g$
if $\JJ$ is not conflicting on $g$.
\end{definition}

Next we can start to define that
$\HH$ is $\emph{compatible}$ on a gadget $g$
as follows:
\begin{definition}
Given a hierarchal structure $\HH=(V(\HH), E(\HH))$ and a gadget $g\in \tilde{\VV}$,
we say $\HH$ is \emph{compatible} on $g$ if
\begin{enumerate}
\item there exists a directed path $P\in \HH(e_C)$ such that
for each vertex $(l, g)\in P$,
each vertex $(l', g')\in P$ other than $(l, g)$
is compatible to $(l, g)$ on $g$. And,
\item for each arc $e\in P$,
the hierarchal substructure $\HH_e$ of $\HH$ is also compatible on $g$.
\end{enumerate}
We say (1) a directed path $P\in \HH(e_C)$ is \emph{compatible} on $g$
if $P$ satisfies the conditions 1 and 2. And,
(2) a directed path $P\in \HH(e_C)$ is \emph{conflicting} on $g$
if $P$ violates the condition 1 or the condition 2.
Moreover, an arc $e\in P$ is \emph{compatible}
on $g$ if
$\HH_e$ is compatible on $g$.
On the other hand, $e$ is \emph{conflicting} on $g$ if
$\HH_e$ is not compatible on $g$.
\end{definition}

From the above definition of a compatible hierarchal structure $\HH$,
we immediately have a recursive procedure to check whether
there exists a compatible path $P$ on $g$ in $\HH(e_C)$ as follows:
for each directed path $P\in \HH(e_C)$,
recursively check each arc $e\in P$ whether the DAG $\HH(e)\in \HH_e$
($\HH(e)$ is the root's associated DAG in $\HH_e$)
has a compatible directed path on $g$ or not.
Then $P$ is conflicting on $g$ if and only if $P$ becomes disconnected
after removing all conflicting arcs $e$ on $g$ from $\HH(e_C)$.
It is stated in Property \ref{prop:connected-v-chain}.

Briefly speaking, Algorithm \ref{alg:backtrack-conflict-gadgets} iteratively
removes the root $\JJ(r)$ of a conflicting hierarchal-v-chain $\JJ$ from $\HH(e_C)$.
Also, we utilize Algorithms \ref{alg:EXPAND} and
\ref{alg:MERGE} to adjust the structure of $\HH$.
Moreover, after recursively adjusting $\HH$
(it means that via adjusting $\HH(e_C)\in \HH$, we also
adjust the structure of $\HH_e, e\in \HH(e_C)$), we have the following fact:
for each $\HH(e)\in \HH$, if there does not exist
a directed path between $v_1$ and $v_2$ in $\HH(e)$ before removing $g$ from $\HH(e)$,
then $v_1$ remains disconnected to $v_2$ in $\HH(e)$ after removing $g$ from $\HH(e)$.
At the end of Algorithm \ref{alg:backtrack-conflict-gadgets},
we can conclude that
each connected path $P\in \HH(e_C)$
has a corresponding compatible hierarchal v-chain $\JJ$.
The intuitive meaning of a path $P\in \HH(e_C)$
keeps its connectivity after removing $g$
is that
$P$ can add $g$ into its corresponding hierarchal v-chain $\JJ$.

How to efficiently
check whether a directed path $P \in \HH(e_C)$ is compatible on $g$ or not?
We can recursively
check whether there exists a compatible directed path $P'\in \HH(e)$ on $g$
for each complete and partial arcs $e\in \HH(e_C)$.
In Observations \ref{obs:complete-edge} and \ref{obs:partial-edge},
we describe the recursive formulas to
check complete arcs and partial arcs
whether they are compatible on $g$ or not.
After we check all complete arcs and partial arcs,
we keep all compatible arcs on $g$ in $\HH(e_C)$ and
check whether there exists a directed path
from source to sink in $\HH(e_C)$.
(The root $\JJ(r)$ of a compatible hierarchal v-chain $\JJ$.)
The recursive procedures to check
directed paths, complete arcs and partial arcs on $g$ in $\HH(e_C)$
are described
in Lemmas \ref{lemma:proper-chain}, \ref{lemma:complete-edge} and \ref{lemma:partial-edge},
respectively.

\begin{algorithm}[t]
\caption{Find Conflicting Gadgets via Backtracking Algorithm}
\label{alg:backtrack-conflict-gadgets}

\KwIn{Sets $\tilde{\VV}$ and $\hat{\VV}$}

Add the backbone graph $\HH_{C=\beta(g_T)}$ of $g_T$ into $\HH$.
$V(\HH)=\{\HH(e_C)=\HH_C\}$\;

\While{the initial fan $g_0$ is not removable from $\HH(e_C)$}
{
Find a removable gadget $g$ from $\HH(e_C)$\;

\uIf {$g$ is a mirror fan}
{
\For{ each $\MM$-triple $(g^B, g, g^U) \in \hat{\VV}$}
{
EXPAND $g$ in $\HH$ by Algorithm \ref{alg:EXPAND}\;
}
}
\ElseIf{$g$ is a fan}
{
\For{ each $\GG$-pair $(g, g^R)$}
{
Recursively check whether
each complete arc and partial arc in $\HH(e_C)$ are compatible on $g$ or not
by Algorithm \ref{alg:MERGE}\;

}
Recursively delete all complete arcs on $g$ from $\HH(e_C)$\;
}

Remove $g$ from $\tilde{\VV}$ and
all connections $(g, g^R)$ and $(g^B, g, g^U)$ from $\hat{\VV}$\;

}

\uIf {there exists a path $P=(e_1, e_2, \cdots, e_k) \in \HH(e_C)$
where each $e_i, 1\leq i\leq k,$
is a complete arc on $g_0$}
{
$G$ have an area-universal rectangular layout\;
}
\Else
{
$G$ does not have any area-universal rectangular layout\;
}
\end{algorithm}

The main task for EXPAND operation in Algorithm \ref{alg:EXPAND} is to add the backbone graph of a mirror fan into $\HH$.
See Fig \ref{fig:expand} as an example for EXPAND operation.

\begin{algorithm}[t]
\caption{EXPAND a mirror fan in $\HH$}
\label{alg:EXPAND}

\KwIn {The hierarchal structure $\HH$ with the root $\HH(e_C)$ and
a $\MM$-triple $(g^B, g, g^U) \in \hat{\VV}$}

Add the backbone graph $\HH_{L=\beta(g)}$ into $\HH$ as a vertex $v\in V(\HH)$\;

\For{each DAG $\HH(e)\in V(\HH)$ such that
there is a subpath $(l_1, g^B)\xrightarrow{e_1} (\alpha(g), g) \xrightarrow{e_2} (l_2, g^U)$ in $\HH(e)$
where $l_1$ and $l_2$ are portions of $\alpha(g^B)$ and $\alpha(g^U)$, respectively}
{

Replace $(l_1, g^B)\xrightarrow{e_1} (\alpha(g), g) \xrightarrow{e_2} (l_2, g^U)$ by
$(l_1, g^B)\xrightarrow{e'} (l_2, g^U)$ in $\HH(e)$\;

Set $\HH(e') = \HH(v)$ and
add a super arc from $\HH(e)$ to $\HH(v)$ in $\HH$\;

Add an arc from the starting vertex of $e$ to the source vertex of $\HH(v)$ and
an arc from the sink vertex of $\HH(v)$ to the ending vertex of $e$ in $\HH(e)$\;

Remove the vertex $(\alpha(g), g)$ from $\HH(e)$\;
}
\end{algorithm}

In the following lemma,
we describe the recursive structure of a compatible directed path $P\in \HH(e_C)$.
A simple way to explain Lemma \ref{lemma:proper-chain} is that
to recursively check a compatible directed path $P$ in $\HH(e_C)$
is equal to,
for each arc $e\in P$,
recursively check whether there exists a compatible directed path $P'$ in $\HH(e)$.
In general, each directed path $P\in \HH(e_C)$ can be decomposed into five parts:
(1) the subpath from source which doesn't have any partial arc and complete arc on $g$,
(2) the subpath which only has a left partial arc on $g$,
(3) the subpath which only has complete arcs on $g$,
(4) the subpath which only has a right partial arc on $g$ and
(5) the subpath to sink which doesn't have any partial arc and complete arc on $g$.
Because each arc $e$ in a compatible directed path $P$
must be compatible on $g$ ,
it implies that
$\HH(e)$ must have at least one directed path from source to sink which is compatible on $g$.
We describe their recursive structures of complete arcs and partial arcs on $g$ in Observations
\ref{obs:complete-edge}, \ref{obs:minimal-complete-edge}, \ref{obs:partial-edge} and \ref{obs:minimal-partial-edge}.
From the above discussion, we immediately have Lemma \ref{lemma:proper-chain}.

\begin{lemma}\label{lemma:proper-chain}
Given a fan $g$,
suppose there is a directed path
$P=(e_1, e_2, e_z,$ $e^p, e^c_1, e^c_2,$ $\cdots, e^c_{k_c}, e^q,$
$e'_1, e'_2, \cdots, e'_{z'}) \in \HH(e_C)$
where the arcs $e^p$ and $e^q$ are the left and right partial arcs on $g$, respectively, and
each arc $e^c_i, 1\leq i\leq k_c,$ is a complete arc on $g$.
Then, $P$ is the root $\JJ(r)$ of a compatible hierarchal v-chain $\JJ$ on $g$ if and and if
\begin{itemize}
\item for each complete arc $e^c_i\in P, 1\leq i\leq k_c$,
there is a compatible directed path $P^c_i\in \HH(e^c_i)$ on $g$
(see Observations \ref{obs:complete-edge} and \ref{obs:minimal-complete-edge}
for more details of a complete arc),
\item for the left partial arc $e^p\in P$,
there is a compatible directed path $P^p\in \HH(e^p)$ on $g$
(see Observations \ref{obs:partial-edge} and \ref{obs:minimal-partial-edge}
for more details of a left partial arc) and
\item for the right partial arc $e^q\in P$,
there is a compatible directed path $P^q\in \HH(e^q)$ on $g$
(see Observations \ref{obs:partial-edge} and \ref{obs:minimal-partial-edge}
for similar details of a right partial arc).
\end{itemize}
\end{lemma}

Given a complete arc $e=(l_1, g)\rightarrow (l_2, g)\in \HH(e_C)$ on $g$,
each arc $e'\in P$ is also a complete arc on $g$.
And, we know that if we want to guarantee that a complete arc $e$ is compatible, we
must recursively check whether $\HH(e)$ can have a directed path which only consists of
compatible complete arcs on $g$.
Obviously, to recursively check a compatible complete arc on $g$
is a recursive procedure implemented by dynamic programming technique.
Also, the base case for the recursive procedure
is that a complete arc on $g$ whose two end vertices $(l_1, g)$ and $(l_2, g)$
have that $l_1\cup l_2$ is contiguous on the front boundary of $g$. It means that
$(l_1, g) \rightarrow (l_2, g)$ is compatible on $g$.
See Fig \ref{fig:merge} as an example of a complete arc.

From the above discussion, we can describe recursive structures of a compatible
complete arc on $g$ in Observations \ref{obs:complete-edge} and \ref{obs:minimal-complete-edge}:



\begin{observation}\label{obs:complete-edge}
Given a fan $g$, a complete arc
$e$ is compatible on $g$ if and only if
there exists a directed path $P=(e^c_1, e^c_2, \cdots, e^c_{k_c}) \in \HH(e)$
where
\begin{itemize}
\item the source vertex of $P$ is the starting vertex of $e$,
\item the sink vertex of $P$ is the ending vertex of $e$ and
\item each $e^c_i, 1\leq i\leq k_c,$ is a compatible complete arc on $g$.
\end{itemize}
\end{observation}

\begin{observation}\label{obs:minimal-complete-edge}
Given a fan $g$,
a minimal complete arc $e=(l_1, g)\rightarrow (l_2, g)$ is compatible on $g$
if and only if
$l_1\cup l_2$ is contiguous on the front boundary $\alpha(g)$ of $g$
(the last vertex of $l_1$ overlaps the first vertex of $l_2$).
\end{observation}

Based on Observations \ref{obs:complete-edge} and \ref{obs:minimal-complete-edge},
we can check a complete arc $e \in \HH(e_C)$ on $g$ whether it is compatible
on $g$ or not via Lemma \ref{lemma:complete-edge}.

\begin{lemma}\label{lemma:complete-edge}
Given a fan $g$,
we can recursively check whether each complete arc $e\in \HH(e_C)$ on $g$
satisfies structure described in Observations \ref{obs:complete-edge} and \ref{obs:minimal-complete-edge} as follows:
\begin{enumerate}
\item recursively check whether each arc $e'\in \HH(e)$ satisfies the structures in Observations \ref{obs:complete-edge}
and \ref{obs:minimal-complete-edge},
\item keep all arcs passing the above tests in $\HH(e)$, and
\item check whether $\HH(e)$ has a directed path from source to sink. If yes, keep $e$ in $\HH(e_C)$.
Otherwise, delete $e$ from $\HH(e_C)$.
\end{enumerate}
\end{lemma}

For a left partial arc $e=(l_1, g_1)\rightarrow (l_2, g_2)\in \HH(e_C)$ on $g$,
because $g_2$ is equal to $g$,
each directed path $P$ in $\HH(e)$ can be partitioned into
(1) the subpath that consists of neither complete arcs nor partial arcs on $g$,
(2) the left partial arc on $g$ and
(3) the subpath that only consists of complete arcs on $g$.
See Fig \ref{fig:merge} as an example of a left partial arc.
Similarly, to check a compatible left partial arc on $g$ is a recursive procedure
which can be implemented by dynamic programming technique.
Also, the base case for the recursive procedure
is a left partial arc $(l_1, g_1) \rightarrow (l_2, g_2=g)$ on $g$
which has (1) $(g_2=g, g_1)$ is a v-connection in $\hat{\VV}$ and
(2) $l_1\cup l_2$ is contiguous on the front boundary $\alpha(g_2, g_1)$ of the connection $(g_2, g_1)$.
It means that $(l_1, g_1) \rightarrow (l_2, g_2=g)$ is compatible on $g$.
See Fig \ref{fig:merge} for examples of a left partial arc and a minimal left partial arc.
From the above discussion, we can describe recursive structures of a compatible
left partial arc on $g$ in Observations \ref{obs:partial-edge} and \ref{obs:minimal-partial-edge}:

\begin{observation}\label{obs:partial-edge}
Given a fan $g$ and a left partial arc $e$ on $g$,
a left partial arc $e$ is
compatible on $g$ if and only if
there exists a directed path $P=(e_1, e_2, \cdots, e_z, e^p, e^c_1, e^c_2, \cdots, e^c_{k_c})$
in $\HH(e)$ where
\begin{itemize}
\item the source vertex of $P$ is the starting vertex of $e$,
\item the sink vertex of $P$ is the ending vertex of $e$,
\item for each $1\leq i\leq z$, $e_i=(l_i, g_i)\rightarrow (l_{i+1}, g_{i+1})$
is an arc where $g \neq g_i$ and $g\neq g_{i+1}$,
\item the arc $e^p$ is a compatible left partial arc on $g$, and
\item each arc $e^c_i, 1\leq i\leq k_c,$ is a compatible complete arc on $g$.
\end{itemize}
\end{observation}

\begin{observation}\label{obs:minimal-partial-edge}
Given a fan $g$,
a minimal left partial arc $e=(l_1, g_1)\rightarrow (l_2, g_2=g)$ on $g$
is a compatible left partial arc on $g$ if and only if
(1) $(g_2=g, g_1)$ is a connection in $\hat{\VV}$ and
(2) $l_1\cup l_2$ is contiguous on the front boundary $\alpha(g_2, g_1)$ of the connection $(g_2, g_1)$
(the last vertex of $l_1$ overlaps the first vertex of $l_2$).
\end{observation}

Based on Observations \ref{obs:partial-edge} and \ref{obs:minimal-partial-edge},
we can recursively check whether a partial arc $e \in \HH(e_C)$
is compatible on $g$ or not via Lemma \ref{lemma:partial-edge}.

\begin{lemma}\label{lemma:partial-edge}
Given a fan $g$,
we can recursively check whether a partial arc $e\in \HH(e_C)$ on $g$
satisfies the structures in Observations \ref{obs:partial-edge} and \ref{obs:minimal-partial-edge} as follows:
\begin{enumerate}
\item recursively check whether
each partial arc $e'\in \HH(e)$ on $g$
satisfies the structures in Observations \ref{obs:partial-edge} and \ref{obs:minimal-partial-edge},
\item recursively check whether
each complete arc $e'\in \HH(e)$ on $g$ satisfies the structures
in Observations \ref{obs:complete-edge} and \ref{obs:minimal-complete-edge},
\item keep all arcs passing the above tests in $\HH(e)$, and
\item check whether $\HH(e)$ has a directed path from source to sink. If yes, keep $e$ in $\HH(e_C)$.
Otherwise, delete $e$ from $\HH(e_C)$.
\end{enumerate}
\end{lemma}

There are three main tasks of MERGE operation in Algorithm \ref{alg:MERGE}.
The first one is to recursively check each complete arc on a removable gadget $g\in \tilde{\VV}$.
The second one is to recursively check each partial arc on $g$.
The final one is to remove $g$ and maintain the connectivity
for each compatible directed path
in $\HH(e_C)$.
What we do in the final for loop is to
reconnect a new arc between a left partial arc $e_L$ and a right partial arc $e_R$
if and only if there exists a compatible directed path from $e_L$ to $e_R$.
See Fig \ref{fig:merge} as an examples for Algorithm \ref{alg:MERGE}.

\begin{algorithm}[t]
\caption{MERGE $\HH$ via Dynamic Programming}
\label{alg:MERGE}

\KwIn{The hierarchal structure $\HH$ with the root $\HH(e_C)$ and a $\GG$-pair $(g, g^R)\in \hat{\VV}$}

\For {each complete arc $e^c\in \HH(e_C)$ on $g$}
{
Recursively check $e^c$
whether $e^c$ is compatible on $g$ or not
(this recursive check follows Lemma \ref{lemma:complete-edge})\;
}

\For {each partial arc $e^p\in \HH(e_C)$ on $g$}
{Recursively check $e^p$
whether $e^p$ is compatible on $g$ or not
(this recursive check follows Lemma \ref{lemma:partial-edge})\;
}

\For {each pair of left partial arc $e_L=v_a\rightarrow v_L=(v_L, g)$
and right partial arc $e_R=v_R=(v_R, g)\rightarrow v_b$ in $\HH(e_C)$
    such that $v_L$ and $v_R$ remain connected in $\HH(e_C)$}
{

      Change $v_L=(v_L, g)$ and $v_R=(l_R, g)$ to
      $v_L=(\beta(g)\cap \alpha(g^R), g^R)$ and
      $v_R=(\beta(g)\cap \alpha(g^R), g^R)$, respectively\;

      Add an arc $e'=(v_L\rightarrow v_R)$ into the graph $\HH(e_C)$ and
      let $\HH(e')$ be the directed path
      $(v_L \rightarrow (\beta(g)\cap \alpha(g^R), g^R) \rightarrow v_R)$\;
}

\end{algorithm}

Note that the connectivity between $v_L$ and $v_R$ is based on the arcs
which are compatible on $g$
in Algorithm \ref{alg:MERGE}.

There are two important properties for the correctness of Algorithm \ref{alg:backtrack-conflict-gadgets}.
The first one states that
we can eliminate each conflicting directed path $P$ (hierarchal v-chain $\JJ$) on $g$ via removing $g$ from $\HH(e_C)$.
The second one states that
the number of directed paths (hierarchal v-chains) decreases
during Algorithm \ref{alg:backtrack-conflict-gadgets} executes.

\begin{property}\label{prop:connected-v-chain}
For each DAG $\HH(e)\in V(\HH)$,
a directed path $P\in \HH(e)$
is conflicting on $g$ if and only if $P$ becomes disconnected after removing $g$ from $\HH(e)$.
\end{property}

\begin{property}\label{prop:dummy-v-chain}
For each DAG $\HH(e)\in V(\HH)$,
if any two vertices $u, v\in \HH(e)$ are
disconnected, then $u$ and $v$ remain disconnected
after removing $g$ from $\HH(e)$.
\end{property}

From the above two properties, we can see that
if we can recursively guaranteed that
for each compatible arc $e\in \HH(e_C)$ on $g$, $\HH(e)$
has Properties \ref{prop:connected-v-chain} and \ref{prop:dummy-v-chain},
then each compatible directed path $P \in \HH(e_C)$ on $g$
is also a compatible hierarchal v-chain on $g$.
And, in the final "for" loop of Algorithm \ref{alg:MERGE},
it connects a new arc between $v_L$ and $v_R$
if and only if there is a compatible directed path on $g$ from $v_L$ and $v_R$.
Hence
it guarantees that no pair of vertices $(v_L, v_R)$ turns into connected if
$v_L$ and $v_R$ are disconnected
before removing conflicting arcs on $g$.

\begin{theorem}\label{thm:backtracking}
Algorithm \ref{alg:backtrack-conflict-gadgets} can successfully
check whether there exists a directed path $P\in \HH(e_C)$
such that $P$ has a corresponding slant $\REL$.
\end{theorem}

\begin{proof}

The correctness of Algorithm \ref{alg:backtrack-conflict-gadgets}
is based on Properties \ref{prop:connected-v-chain} and
\ref{prop:dummy-v-chain}.
Clearly, from Property \ref{prop:connected-v-chain},
each conflicting directed path $P\in \HH(e_C)$
(the root $\JJ(r)$ of each hierarchal v-chain $\JJ$) on $g$
becomes disconnected after removing a removable gadget $g$
and from Property \ref{prop:dummy-v-chain},
it remains disconnected in the following steps.
Hence in the final step of Algorithm \ref{alg:backtrack-conflict-gadgets},
each connected directed path $P\in \HH(e_C)$ has been proven that
$P$ corresponds to the root $\JJ(r)$ of a compatible hierarchal v-chain on $g$
for every gadget $g\in \tilde{\VV}$.
Also,
if a directed path $P\in \HH(e_C)$, corresponds to the root $\JJ(r)$
of a hierarchal v-chain $\JJ$, keeps its connectivity
after recursively removing a gadget $g$ from $\HH(e_C)$, then
this hierarchal v-chain $\JJ$ can add the gadget $g$ into its corresponding slant $\REL$.
Hence we can have that each connected directed path in $\HH(e_C)$
has a corresponding slant $\REL$ $\R$.

\end{proof}

Theorem \ref{thm:backtracking} has proven that
we can backtrack all directed paths $P\in \HH(e_C)$
to know whether $P$ represents
the chain of a partial slant $\REL$ $\R$.

\begin{theorem}\label{thm:backtracking-time}
The time complexity of Algorithm \ref{alg:backtrack-conflict-gadgets} is polynomial bound.
\end{theorem}

\begin{proof}

Let $\KK$ be the number of iterations in Algorithm \ref{alg:backtrack-conflict-gadgets}
and $N_i, 1\leq i\leq \KK$, be the size of $\HH$ in the $i$-th iteration.
The time analysis of backtracking is based on three parts:
(1) $\KK$ is polynomial bound,
(2) each $N_i, 1\leq i\leq \KK$, is polynomial bound and
(3) time complexity $T(N_i)$ in each $i$-th iteration is polynomial bound.

Obviously, the number $\KK$ of total iterations is bounded by the number of connections $\hat{\VV}$.
By Lemma \ref{lemma:number}, the number of gadgets in $G$ is at most $N=O(n^2)$ and the number
of connections in $G$ is at most $O(N^3)=O(n^6)$.
Hence $\KK$ polynomially grows with respect to the number of $G$'s vertices $n$.

For each $i$-th iteration, we either execute EXPAND or MERGE operations to adjust $\HH$' structure.
When Algorithm \ref{alg:backtrack-conflict-gadgets} executes EXPAND operation on a removable gadget $g$ with the back boundary $L=\beta(g)$,
we add $g$'s backbone graph $\HH_L$ into $\HH$ where $\HH_L$'s size (the number of vertices
in $\HH_L$) is polynomial bound.
Since the number $\KK$ of total iterations is polynomial bound,
the total vertices added into $\HH$ in all EXPAND operations are bounded
by the summation of all backbone graphs's size.
Hence, the summation of all backbone graphs's size is polynomial bound.

When Algorithm \ref{alg:backtrack-conflict-gadgets} executes MERGE operation
on a removable gadget $g$ with a $\GG$-pair $(g, g^R)$,
we replace each maximal compatible directed path $P$ on $g$ in $\HH$ by
an arc $e'=(g_u, l_u)\rightarrow (g_v, l_v)$ between the two end vertices of $P$ where
each arc in $P$ is a complete arc on $g$, and
add $\HH(e')$ into $\HH$ where $|\HH(e')|$ consists of the newly-added vertex $(g^R, l)$.
Note that the added vertex $(g^R, l)$ only connects to the two end vertices $(g_u, l_u)$ and $(g_v, l_v)$ of $e$,
and cannot be connected to other vertices in $\HH$ in following iterations.
Also, $(g^R, l)$ is removed from $\HH$ when removing a removable gadget $g^R$ from $\HH$.
Hence the total number of newly-added arcs $e'$ are summation of all backbone graphs's size and
it is polynomial bound.
Also, the size of all added $\HH(e')$ is polynomial bound.

Because the total vertices added into $\HH$ are polynomial bound during each iteration
and the number $\KK$ of iterations is polynomial bound,
the maximum of $\HH$'s size is polynomial bound of each iteration
in Algorithm \ref{alg:backtrack-conflict-gadgets}.
Hence each number $N_i, i\geq 1$, of vertices of $\HH$ in
each $i$-th iteration is polynomial bound.

Now we analyze the time complexity of each iteration.
When Algorithm \ref{alg:backtrack-conflict-gadgets} executes EXPAND operation on a removable gadget $g$ with the back boundary $L=\beta(g)$,
we take polynomial time to add the $g$'s backbone graph $\HH_L$ into $\HH$
because $\HH_L$'s size is polynomial bound.

When Algorithm \ref{alg:backtrack-conflict-gadgets} executes MERGE operation on a removable gadget $g$,
the tasks of MERGE operation consist of
(1) recursively checking each complete arc on $g$ in $\HH$
by Lemma \ref{lemma:complete-edge},
(2) recursively checking each partial arc on $g$ in $\HH$ by Lemma \ref{lemma:partial-edge} and
(3) check that for each DAG $\HH(e)$,
whether there exists a directed path
from the source to the sink in $\HH(e)$
after removing all conflicting complete and partial arcs from $\HH(e)$.

The total work of a MERGE operation can be simply described as follows:
check each DAG $\HH(e)\in V(\HH)$ whether $\HH(e)$ has at least one connected
directed path from the source to the sink in $\HH(e)$.
And, it can be done by executing a breadth-first search in $\HH(e)$
since $\HH(e)$ is a DAG. Hence the complexity of the total work of a MERGE operation
is polynomial bound as $\mbox{poly}(\max_{1\leq i\leq \KK}N_i)$.
Note that the order to check each DAG $\HH(e)\in V(\HH)$
is a bottom-up traversal in $\HH$ as follows:
a DAG $\HH(e)\in V(\HH)$ is ready to check if and only if
every DAG $\HH(e'), e'\in E(\HH(e))$, has been checked.


Because each iteration takes polynomial time as $T(N_i)\leq \mbox{poly}(\max_{1\leq i\leq \KK}N_i)$,
the total time complexity of all $\KK$ iterations in Algorithm \ref{alg:backtrack-conflict-gadgets}
is also polynomial bound.

\end{proof}

Theorem \ref{thm:backtracking-time} has proven that
the time complexity of Algorithm \ref{alg:backtrack-conflict-gadgets}
is polynomial bound.
\clearpage

\bibliographystyle{plain}

\end{document}